\documentclass[aps,pra,onecolumn,11pt,superscriptaddress]{revtex4-2}

\usepackage{silence}
\usepackage{graphicx} %
\usepackage{amsmath, amsthm, amssymb, mathtools}%
\usepackage{hyperref}
\usepackage[nameinlink]{cleveref}%
\usepackage[shortlabels]{enumitem}
\usepackage[a4paper, margin=0.75in]{geometry}%
\usepackage{multirow}
\usepackage{xcolor} %

\usepackage{thmtools}
\declaretheorem[numberwithin=section]{theorem}
\declaretheorem[sibling=theorem]{
    corollary,
    conjecture,
    proposition,
    lemma,
    definition,
    example,
    remark
}

\DeclarePairedDelimiter\ket{\lvert}{\rangle}
\DeclarePairedDelimiterX\braket[2]{\langle}{\rangle}{#1 \delimsize\vert #2}
\DeclarePairedDelimiterX\ketbra[2]{\delimsize\vert}{\delimsize\vert}{#1 \rangle \langle #2}
\newcommand{\N}{\mathbb{N}}
\newcommand{\Z}{\mathbb{Z}}

\newcommand{\R}{\mathbb{R}}

\newcommand{\Q}{\mathbb{Q}}
\newcommand{\C}{\mathbb{C}}

\newcommand{\inv}{^{-1}}

\newcommand{\dagg}{^\dag}

\newcommand{\kz}{\ket{0}}
\newcommand{\ko}{\ket{1}}

\newcommand{\kp}{\ket{+}}
\newcommand{\km}{\ket{-}}

\newcommand{\kpsi}{\ket{\psi}}

\newcommand{\onesqrttwo}{\frac{1}{\sqrt{2}}}
\newcommand{\onesqrtn}{\frac{1}{\sqrt{n}}}
\newcommand{\onesqrtm}{\frac{1}{\sqrt{m}}}
\newcommand{\onesqrtq}{\frac{1}{\sqrt{q}}}

\newcommand{\allones}{\mathbf{1}}
\newcommand{\allonesn}{\allones_n}
\newcommand{\allonesm}{\allones_m}
\newcommand{\allonesq}{\allones_q}

\newcommand{\given}{\;\middle|\;}
\newcommand{\mufz}{\{0, \pm1, \pm i\}}%
\newcommand{\omegvec}{\binom{\omega_1}{\omega_2}}
\newcommand{\omegvecvavbspace}{\begin{pmatrix} \omega_1 V_A\\[1mm] \omega_2 V_B \end{pmatrix}}
\newcommand{\omegvecvavb}{\begin{pmatrix} \omega_1 V_A\\ \omega_2 V_B \end{pmatrix}}
\newcommand{\omegvecvavbT}{\begin{pmatrix} \omega_1 V_A,\, \omega_2 V_B \end{pmatrix}^T}

\newcommand{\mcsym}{M_{\mathrm{csym}}}
\newcommand{\mrsym}{M_{\mathrm{rsym}}}
\newcommand{\rasym}{R_{\mathrm{asym}}}
\newcommand{\subspace}{\mathcal S}

\usepackage{xspace}

\DeclareMathOperator{\sign}{sgn}
\DeclareMathOperator{\tr}{Tr}%
\DeclareMathOperator{\Span}{span}%
 
\begin{document}

\title{Universal Complex Quantum-Like Bits from Hermitian Weighted Graphs}

\author{Ethan Dickey}
\affiliation{Department of Computer Science, Purdue University, West Lafayette, Indiana 47907, USA}
\author{Sabre Kais}
\affiliation{Department of Electrical and Computer Engineering, North Carolina State University, Raleigh, North Carolina 27606, USA}
\affiliation{Department of Chemistry and Department of Physics, North Carolina State University, Raleigh, North Carolina 27606, USA}

\begin{abstract}
    We study when block-coupled regular graphs can realize prescribed complex quantum-like (QL) bit states as exact synchronized eigenstates. Two regular subgraphs $G_A$ and $G_B$ supply normalized all-ones eigenvectors $V_A$ and $V_B$, and algebraically regular bipartite couplings reduce the full graph-supported operator exactly to a $2\times 2$ effective block on $\subspace=\Span \{ \kz, \ko \}$. Within this reduction we prove that two natural symmetric complexifications are not universal for state synthesis under a real-spectrum requirement: complex symmetric coupling with real diagonal regularities forces the target computational basis amplitude ratio $r=\omega_2/\omega_1$, for $\kpsi = \omega_1\kz + \omega_2\ko$, to satisfy $r^2\in\R$, while real symmetric coupling with complex diagonal regularities forces $r+1/r\in\R$. Replacing complex symmetry by Hermitian coupling removes this phase obstruction. For any nonbasis target state, any prescribed real eigenvalue, and any prescribed nonzero signed spectral gap, a Hermitian weighted coupling realizes the target exactly. Additionally, an independently tuned directed-coupling model gives a second universality mechanism. We then pass from continuous effective parameters to finite weighted graphs with entries in $\mufz$ (the fourth roots of unity and zero), characterize the balanced discrete coupling lattice by perfect matchings, and show that exact discrete Hermitian realizations are dense in the synchronized pure-state space. These results give a state-synthesis universality taxonomy for complex QL-bits and identify Hermitian conjugate pairing as the robust structural mechanism that supports arbitrary complex amplitudes with real two-level spectra.

\end{abstract}
\maketitle

\section{Introduction}
Many large systems are too complicated to understand or control one component at a time. Their useful behavior, however, is often concentrated in a small number of collective modes, even when the underlying state space is enormous. Reduced-order modeling exploits this separation by replacing large dynamical models with compact descriptions \cite{quarteroni2014reduced,benner2015survey}, while graph signal processing uses graph-supported operators to describe collective structure in data on irregular domains \cite{shuman2013emerging,ortega2018graph}. In this paper, we treat the graph-supported operator itself as a design object. We ask whether it can be constructed so that prescribed collective modes exist exactly and serve as information-bearing states.

Quantum information yields a precise way to test this because its computational states are vectors in a complex Hilbert space. In particular, a pure qubit state is specified, up to a global phase, by the relative magnitude and phase of two amplitudes, meaning the existence of two distinguished modes in a network only captures part of the structure of the system. Representing the full qubit state space requires arbitrary complex superpositions within a persistent two-dimensional sector. Studying that requirement in a classical network lets us separate the geometry of the state space from the additional resources needed for gates, measurement, entanglement, and scalable quantum computation. Our goal here is to determine when a constrained network can realize every complex two-level state exactly. Throughout the paper, \textit{universal} refers only to this state-synthesis task. This notion does not imply a universal gate set or a quantum computational advantage.

Recent work on quantum-like (QL) states has shown that synchronized classical networks can contain exact low-dimensional sectors whose states admit a qubit-like description \cite{scholes2024quantumlike,amati2025quantumlikebits,scholes2025productstates}. In these constructions, collective modes supply the basis states, while graph regularities and couplings determine how those modes combine. Subsequent work gave an exact signed-graph construction for arbitrary real single QL-bit states \cite{dickey2025construction}. However, a general two-level quantum state also carries a relative phase. Complex weights can encode that phase, but they can also produce complex eigenvalues. If the graph-supported operator is used to generate continuous-time evolution, an imaginary part in an eigenvalue causes the corresponding mode to grow or decay rather than evolve only by a phase. Complex state synthesis and real spectral realization therefore must be addressed together.

Viewed computationally, this is a constrained synthesis problem for an exact low-dimensional representation inside a much larger structured system. Graph-supported operators already generate computationally meaningful dynamics in continuous-time quantum walks, state-transfer models, and graph-encoded computation \cite{farhi1998quantum,kempe2003quantum,bose2003communication,christandl2004perfect,christandl2005perfect,childs2009universal,godsil2012state}. Equitable partitions and quotient graphs (partitions for which every vertex in a given section has the same total, possibly weighted, adjacency into each other section, so that the adjacency action on section-constant vectors is represented exactly by a smaller quotient matrix) show how regularity can induce exact low-dimensional invariant dynamics inside a larger graph \cite{krovi2007quantum,bachman2012perfect,aguiar2018synchronization}, while structured inverse eigenvalue problems ask how an operator can be constructed to realize prescribed spectral data \cite{chu2002structured}. Our setting combines these viewpoints. The graph-supported operator must create an exact two-level sector, realize a prescribed state within it, and satisfy restrictions on symmetry and admissible edge weights. These constraints are enforced at the level of the full operator, so the target state is an exact eigenstate of that operator rather than only an eigenstate of an approximate effective model.

This common architecture is shown in \Cref{fig:graphandspectrum_adj_msubgraphNone}. Two regular subgraphs provide the collective modes, and a structured bipartite coupling determines how those modes interact. \Cref{sec:prelims} establishes the exact two-dimensional reduction used throughout the paper. Keeping this architecture fixed isolates the effect of coupling symmetry on expressivity.

Relaxing the setting to allow complex edge weights broadens the design space but also raises an important spectral issue. Non-Hermitian and complex-symmetric operators can possess real spectra, but the reality of those spectra is delicate and can depend on hidden structural constraints, antilinear symmetries, or particular parameter regimes \cite{bender1998real,bender2002complex,mostafazadeh2002pseudo,elganainy2018nonhermitian,ashida2020nonhermitian,bergholtz2021exceptional}. These symmetric extensions thus cannot be dismissed in advance. We study two such extensions. The first allows complex-symmetric coupling while keeping the diagonal regularities real; the second keeps the coupling real and allows complex diagonal regularities. We compare both with Hermitian conjugate coupling and with independently tunable directed coupling. Hermitian coupling guarantees a real spectrum, but the full comparison identifies the structural feature that removes the phase restriction. The symmetric models retain reciprocal coupling, the Hermitian model adds conjugate pairing, and the directed model allows the two coupling directions to vary independently. In physical terms, conjugate pairing states that reciprocal edges carry opposite phases, as for hopping amplitudes in a synthetic gauge field \cite{lieb1993fluxes,elganainy2018nonhermitian}, whereas the directed model corresponds to nonreciprocal coupling.

Our first results give exact obstructions for the two symmetric extensions. Let $r$ denote the ratio of the two basis amplitudes of a nonbasis target state. With complex-symmetric coupling and real diagonal regularities, a real two-level spectrum is possible exactly when $r^2\in\R$. With real symmetric coupling and complex diagonal regularities, it is possible exactly when $r+r^{-1}\in\R$. At $r=\pm i$, both symmetric models force the two synchronized eigenvalues to coincide, so these points admit only zero-gap realizations. Thus, each symmetric extension realizes a nontrivial family of complex states, but neither realizes all of them with a prescribed nonzero gap.

Hermitian conjugate coupling removes this phase obstruction. For every nonbasis target state, every prescribed real eigenvalue, and every prescribed nonzero signed separation from the other synchronized eigenvalue, we give an exact Hermitian construction. The synchronized sector is also a reducing subspace of the full operator, so states initialized in that sector do not mix with the remaining network modes. Independently tunable directed coupling provides a second universal construction, even when the two diagonal regularities are constrained to be equal. It achieves the same state-synthesis expressivity, but without the orthogonal decomposition and spectral guarantees supplied by Hermiticity.

The canonical one-qubit $H$- and $T$-type magic states provide familiar checkpoints for this classification
\cite{bravyi2005universal,veitch2013efficient,veitch2014resource,howard2014contextuality,bravyi2012lowoverhead,jones2013multilevel,bravyi2016trading,campbell2017roads,seddon2019quantifying}. The $H$ state is realizable in all four models, whereas the $T$ state is excluded only from the complex-symmetric model with real diagonal regularities. We use these states only as benchmarks for representational power; the present work does not construct Clifford operations, magic-state distillation, or fault-tolerant quantum computation.

The continuous results treat the effective graph parameters as tunable scalars. We therefore also ask which states can be realized by finite weighted graphs with edge alphabet $\mufz$. In the balanced Hermitian setting, we characterize the available complex coupling sums using disjoint perfect matchings. This gives an exact finite realization for every Gaussian-rational amplitude ratio. As the graph size varies, these exactly realizable states are dense, modulo global phase, in the projective synchronized pure-state space.

Together, these results isolate the sources of expressivity in complex QL-bit constructions. Regularity creates the exact low-dimensional sector. Conjugate pairing or independent directionality removes the restriction on relative phase. In the balanced Hermitian setting, the finite edge alphabet determines which continuous parameters admit exact finite realizations. Increasing the graph size cannot remove an obstruction imposed by the symmetry class itself. These results establish the state-synthesis layer of the QL-bit program and provide a foundation for studying graph-native transformations, multibit composition, perturbation robustness, measurement, and the boundary between quantum-like classical networks and quantum computation.

\begin{figure}[ht!]%
    \centering
    \includegraphics[width=\textwidth]{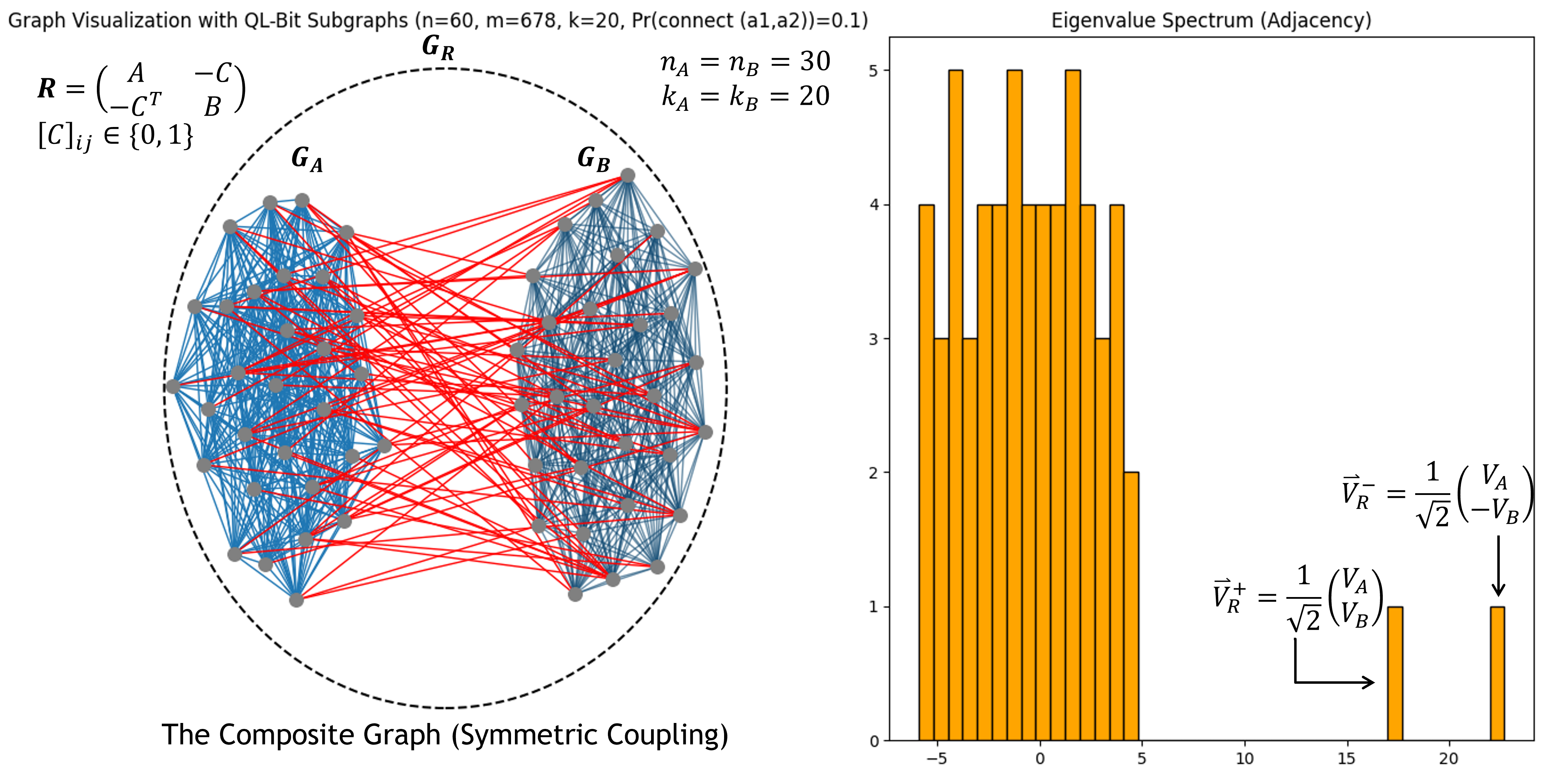}
    \caption{\it A QL-bit built from two regular subgraphs coupled by a bipartite block. The displayed real-symmetric example is the historical random benchmark of \cite{dickey2025construction}; the results below replace its approximate coupling regularity by exact algebraic regularity. Under the displayed sign convention, the symmetric and antisymmetric synchronized combinations $(V_A,V_B)^T/\sqrt{2}$ and $(V_A,-V_B)^T/\sqrt{2}$ have eigenvalues approximately $k\mp l$.}
    \label{fig:graphandspectrum_adj_msubgraphNone}
\end{figure}

\section{Preliminaries}\label{sec:prelims}
Following prior QL-bit work, we begin with two synchronized graph modes whose span supplies the two-level state space used throughout the paper. In our setting, these modes are supported on two coupled subgraphs $G_A$ and $G_B$. The regularity conditions below make their span invariant under the full graph-supported operator; spectral separation from the complementary modes is a separate condition needed when the synchronized eigenstates must be uniquely identifiable in the full spectrum. This section fixes the common graph-theoretic setup.

We work at the adjacency-matrix level. Let $G_A$ and $G_B$ have adjacency matrices $A$ and $B$, where $G_A$ is $k_A$-regular on $n$ vertices and $G_B$ is $k_B$-regular on $m$ vertices. For these real, nonnegative regular adjacency matrices, the normalized all-ones vector is a top eigenvector of each subgraph, and hence a Perron-Frobenius eigenvector \cite{pillai2005perron},
    \begin{equation}\label{eqn:vavb_eigvecs}
        V_A=\onesqrtn\,\allonesn,\qquad V_B=\onesqrtm\,\allonesm
    \end{equation}
with eigenvalues
    \begin{equation}\label{eqn:vavb_eigvals}
        AV_A=k_A V_A,\qquad BV_B=k_B V_B .
    \end{equation}
We regard $V_A$ and $V_B$ as the \emph{synchronized modes} of the two subgraphs. Here ``synchronized'' means that the amplitude is constant across all vertices of the corresponding subgraph. Regularity makes this block-constant vector invariant under the internal adjacency: because every row of $A$ has sum $k_A$, the action of $A$ on $V_A$ only rescales $V_A$, and similarly for $B$ and $V_B$. Thus the internal graph dynamics do not create variation within either subgraph when the system is initialized on these modes. Once the coupling block also has constant weighted row- and column-sums, the full operator maps block-constant vectors back into the two-dimensional span generated by $V_A$ and $V_B$. These vectors are the building blocks from which QL-bit states are assembled.

Embed these into the full space by defining
    \begin{equation}\label{eqn:synchronized_subspace}
        \kz \coloneqq \begin{pmatrix}V_A\\ \mathbf{0}_m\end{pmatrix},\qquad
        \ko \coloneqq \begin{pmatrix}\mathbf{0}_n\\ V_B\end{pmatrix},\qquad
        \subspace \coloneqq \Span\{\kz,\ko\}.
    \end{equation}
Equivalently, the Hadamard basis on $\subspace$ is
    \begin{equation}\label{eqn:kp_km}
        \psi_+\equiv\kp\coloneqq\onesqrttwo(\kz+\ko)=\onesqrttwo\begin{pmatrix}V_A\\ V_B\end{pmatrix},\qquad
        \psi_-\equiv\km\coloneqq\onesqrttwo(\kz-\ko)=\onesqrttwo\begin{pmatrix}V_A\\ -V_B\end{pmatrix}.
    \end{equation}
Thus any normalized synchronized state may be written as
    \[
        \psi=a\,\kp+b\,\km=\omega_1\kz+\omega_2\ko,\qquad
        |a|^2+|b|^2=|\omega_1|^2+|\omega_2|^2=1,
    \]
where
    \[ \omega_1=\onesqrttwo(a+b),\qquad \omega_2=\onesqrttwo(a-b). \]

For a nonbasis state, it is convenient to encode the relative magnitude and phase of its synchronized amplitudes by the projective ratio
    \begin{equation}\label{eqn:projective_ratio}
        r\coloneqq\frac{\omega_2}{\omega_1}\in\C^\times \coloneqq \C\setminus\{0\}.
    \end{equation}
Normalization then determines the state up to global phase:
    \begin{equation}\label{eqn:state_from_ratio}
        \psi=e^{i\theta}\frac{\kz+r\ko}{\sqrt{1+|r|^2}},\qquad \theta\in\R.
    \end{equation}
This expression is a reparameterization (not a transformation of the state, change of basis, or physical nonunitary operation): factoring out $\omega_1$ leaves the basis unchanged, while normalization fixes $|\omega_1|$ once $r$ is specified, and the phase of $\omega_1$ appears only as the global phase $e^{i\theta}$. The synthesis problem studied below is therefore to determine which choices of coupling and detuning make the state specified by a prescribed ratio $r$ an eigenvector of the full operator.

In the original real-valued QL-bit construction \cite{dickey2025construction}, the two subgraphs are coupled by a bipartite block $C$, giving
    \[ R_{\mathrm{real}}=\begin{pmatrix}A&C\\ C^T&B\end{pmatrix}. \]
The complex-weighted extensions considered below keep the same diagonal blocks but vary the relation between the off-diagonal blocks: \Cref{sec:complex_sym:sym_coupling_overview,sec:complex_sym:why_sym_not_work} use
    \[ \begin{pmatrix}A&-C\\ -C^T&B\end{pmatrix}, \]
\Cref{sec:complex_sym:hermitian} uses
    \[ H\coloneqq\begin{pmatrix}A&-C\\ -C\dagg&B\end{pmatrix}, \]
and \Cref{sec:complex_asym} uses
    \[ \rasym\coloneqq\begin{pmatrix}A&-C_A\\ -C_B&B\end{pmatrix}. \]
What matters in every case is the induced action on $\subspace$.

\subsection{Algebraic regularity of the coupling block}\label{sec:complex_sym:regularity}
For complex edge weights, combinatorial degree is no longer the right notion of regularity, because phase cancellations can make a combinatorially regular block act nonuniformly on the synchronized mode. The correct condition is constant weighted row- and column-sums.

We call $C\in\C^{\,n\times m}$ \emph{algebraically regular} if
\begin{equation}\label{eqn:alg_l_regular}
    C\allonesm=s_A\,\allonesn,\qquad C^T\allonesn=s_B\,\allonesm
\end{equation}
for some $s_A,s_B\in\C$. These scalars are not independent: both equations encode the total weighted sum of the entries of $C$, so necessarily
    \[ ns_A=ms_B. \]
Relative to \eqref{eqn:vavb_eigvecs}, this means that $C$ acts by a single effective scalar on the synchronized sector:
    \[ CV_B=l\,V_A,\qquad C^T V_A=l\,V_B,\qquad l\coloneqq s_A\sqrt{\frac{n}{m}}=s_B\sqrt{\frac{m}{n}}. \]
If the lower off-diagonal block is $C\dagg$ instead of $C^T$, then the same computation gives
    \[ C\dagg V_A=l^* V_B. \]
Thus algebraic regularity is precisely the condition that a coupling block preserve synchronization. In the discrete constructions of \Cref{sec:discrete_and_density}, writing $l=c+di$ records the net real and imaginary weight seen by the synchronized mode.

\subsection{Synchronized subspace and effective reduction}\label{sec:prelims:synchronized_reduction}
Once the diagonal blocks preserve $V_A$ and $V_B$, and the off-diagonal blocks act on them by scalar multiplication, the full $(n+m)$-dimensional problem closes on the two-dimensional space $\subspace$.

\begin{proposition}[General synchronized reduction]\label{prop:general_synchronized_reduction}
    Let
        \[ R\coloneqq\begin{pmatrix}A&-X\\ -Y&B\end{pmatrix}, \]
    where $A\in\C^{\,n\times n}$, $B\in\C^{m\times m}$, $X\in\C^{\,n\times m}$, and $Y\in\C^{m\times n}$. Assume
        \[ AV_A=k_A V_A,\qquad BV_B=k_B V_B,\qquad XV_B=l_A V_A,\qquad YV_A=l_B V_B \]
    for some $k_A,k_B,l_A,l_B\in\C$. Then $\subspace$ is $R$-invariant, and relative to the ordered basis $\{\kz,\ko\}$ the restriction $R|_{\subspace}$ is represented by
        \[ M_{\subspace}\coloneqq\begin{pmatrix}k_A&-l_A\\ -l_B&k_B\end{pmatrix}. \]
    Equivalently, for any $\omega_1,\omega_2\in\C$,
    \[
        R\omegvecvavb=
        \begin{pmatrix}
            (k_A\omega_1-l_A\omega_2)V_A\\
            (k_B\omega_2-l_B\omega_1)V_B
        \end{pmatrix}.
    \]
    In particular, $\omegvecvavb$ is an eigenvector of $R$ with eigenvalue $\lambda$ if and only if
        \[ (k_A-\lambda)\omega_1=l_A\omega_2,\qquad (k_B-\lambda)\omega_2=l_B\omega_1. \]

    \begin{proof}
        For any $\omega_1,\omega_2\in\C$,
        \[
            R\omegvecvavb=
            \begin{pmatrix}
                \omega_1 AV_A-\omega_2 XV_B\\
                -\omega_1 YV_A+\omega_2 BV_B
            \end{pmatrix}
            =
            \begin{pmatrix}
                (k_A\omega_1-l_A\omega_2)V_A\\
                (k_B\omega_2-l_B\omega_1)V_B
            \end{pmatrix}\in\subspace.
        \]
        This proves invariance and identifies the coordinate matrix of $R|_{\subspace}$. The final statement is exactly the eigenvalue equation for $M_{\subspace}$.
    \end{proof}
\end{proposition}

This proposition is the common engine behind the rest of the paper. In the symmetric settings, $X=C$ and $Y=C^T$, so algebraic regularity yields $l_A=l_B=l$. In the Hermitian setting, $X=C$ and $Y=C\dagg$, so $l_A=l$ and $l_B=l^*$. In the asymmetric settings, $X=C_A$ and $Y=C_B$, so $l_A$ and $l_B$ can be tuned independently; \Cref{sec:complex_asym} first imposes the common regularity $k_A=k_B$, while \Cref{sec:generalized_complex} subsequently allows independent diagonal detuning. The question in \Cref{sec:complex_sym,sec:complex_asym,sec:generalized_complex} is therefore not whether the reduction exists, but which target states in $\subspace$ and which real two-level spectra these different reduced blocks can realize.

The ratio $r$ is defined relative to the fixed ordered basis $\{\kz,\ko\}$ in \Cref{eqn:synchronized_subspace}. In the Hermitian reduction, a blockwise basis rephasing $\ko\mapsto e^{i\varphi}\ko$ transforms the coordinate ratio and effective coupling according to
    \[ r\mapsto e^{-i\varphi}r,\qquad l\mapsto e^{i\varphi}l, \]
while leaving $lr$ and the spectrum unchanged. In gain-graph language, this is diagonal-unitary switching \cite{reff2012spectral}. Accordingly, the realizability results compare the model classes in the fixed computational basis and under the stated reciprocity, symmetry, and edge-alphabet constraints. A general rephasing may move a realization outside one of those constrained classes.

\section{Complex Edge Weights in Symmetric Coupling and Detuning}\label{sec:complex_sym}
Extending the real-valued construction of \cite{dickey2025construction}, we now allow complex coupling and detuning parameters and determine which target states remain compatible with a real spectrum on the synchronized two-dimensional subspace. Throughout this section, $k_A$, $k_B$, and $l$ are treated as continuously tunable effective regularities under the stated structural restrictions. Explicit full weighted-operator lifts are collected in \Cref{rem:continuous_regularity_construction}; finite-alphabet and finite-graph constraints are imposed in \Cref{sec:discrete_and_density}.

\subsection{Symmetric coupling and the global block matrix.}\label{sec:complex_sym:sym_coupling_overview}
We consider two disjoint undirected regular subgraphs $G_A$ and $G_B$ with adjacency blocks $A$ and $B$ and allow \emph{detuning} $k_A\neq k_B$. We couple them via an undirected bipartite block $C$. To match the eigenvalue relations used throughout this work (cf.\ the form $\lambda_R = k_A - r l$), we absorb a global sign into the off-diagonal blocks and define the complex-symmetric matrix
\begin{equation}\label{eqn:R_def}
    R \coloneqq \begin{pmatrix} A & -C \\[1mm] -C^T & B \end{pmatrix}.
\end{equation}
Because $R^T=R$ but generally $R\dagg \neq R$, the model is complex symmetric but not Hermitian; we therefore do not invoke Hermitian spectral theorems in this section.

Applying \Cref{prop:general_synchronized_reduction} with $X=C$ and $Y=C^T$, the synchronized restriction of $R$ is
    \[ M_{\mathrm{sym}}= \begin{pmatrix} k_A & -l\\ -l & k_B \end{pmatrix}. \]
Thus, in the symmetric setting, the problem is not the existence of the synchronized reduction, but rather which target states and real spectra this reduced block can realize.

\subsection{Why symmetric complexifications are insufficient}\label{sec:complex_sym:why_sym_not_work}

We now separate two requirements for a viable QL-bit construction on the synchronized subspace. The first is \emph{state synthesis}: realizing a prescribed target state
    \[ \kpsi = \omegvecvavbspace, \qquad |\omega_1|^2+|\omega_2|^2=1. \]
The second is \emph{spectral reality and separation}: keeping the corresponding effective two-level spectrum real and separated.
The reality condition is essential because, given a matrix and its eigenvector
    \[ RV_R = \lambda_RV_R \,,\]
if we take $R$ as the generator of a continuous time evolution via
    \[ \psi(t) = e^{-itR}\psi(0) \,, \]
and take $\psi(0) = V_R$, then
    \[ \psi(t) = e^{-itR}V_R = e^{-it\lambda_R}V_R \,.\]
It is desirable for $V_R$ to only be rotated by a phase. However, if the corresponding eigenvalue is complex, say $\lambda_R = \alpha + \beta i$ with $\alpha, \beta \in \R$, we obtain
    \[ e^{-it\lambda_R}V_R = e^{-it\alpha}e^{t\beta}V_R \,. \]
The factor $e^{-it\alpha}$ gives the expected phase rotation, but the nonzero imaginary part $\beta$ produces exponential growth or decay of the model amplitude. Accordingly, throughout this subsection we require the relevant effective eigenvalues to be real, and, when possible, strictly separated by a positive gap.

In particular, the requirement that the relevant eigenvalues be real is not merely a formal convenience. In the present setting, the reduced $2\times2$ synchronized block is intended to encode a quantum-like two-level sector, so its eigenvalues play the role of effective frequencies or energies. Real eigenvalues therefore correspond to stationary oscillatory modes, while complex eigenvalues introduce built-in amplification or dissipation. From this perspective, allowing complex target amplitudes is not itself problematic; the issue is whether those amplitudes can be embedded into a reduced operator whose eigenmode dynamics remain phase-like rather than exponentially unstable. This is the reason for separating the problem of state synthesis from the problem of spectral realization: a construction may reproduce a desired eigenvector algebraically, yet still fail to provide a physically meaningful or stable two-level model if the associated spectrum is not real. Here, stability is meant in this limited spectral sense: real eigenvalues exclude exponential growth or decay of the corresponding eigenmodes. For the separated $2\times 2$ spectra considered below, diagonalizability is automatic and the evolution of arbitrary superpositions remains uniformly bounded in time, although exact norm preservation for all states in the standard inner product requires the Hermitian structure introduced in \Cref{sec:complex_sym:hermitian}.

For nonbasis states, i.e. $\omega_1\omega_2\neq 0$, we use the projective amplitude ratio $r=\omega_2/\omega_1\in\C^\times$ defined in \Cref{eqn:projective_ratio}. By \Cref{eqn:state_from_ratio}, this ratio uniquely determines the normalized state up to global phase. As in \Cref{sec:prelims:synchronized_reduction}, the block constructions considered below therefore reduce to $2\times 2$ scalar systems in the variable $r$. The next two propositions show that neither natural symmetric placement of the complex degrees of freedom yields arbitrary complex target states under a real-spectrum constraint.

\begin{proposition}[Obstruction for complex symmetric coupling with real detuning] \label{prop:symmetric_complex_coupling_obstruction}
    Fix $r\in\C^\times$ and consider the effective reduced matrix
    \begin{align}
        \mcsym \coloneqq
            \begin{pmatrix}
                k_A & -l\\
                -l & k_B
            \end{pmatrix}, \label{eqn:mCsym} \\
        k_A,k_B\in\R,\quad l\in\C^\times.
    \end{align}

    If $\omegvec$, with $r=\omega_2/\omega_1$, is an eigenvector of $\mcsym$ with a real eigenvalue, then %
        \begin{equation}\label{prop:symm_obstruction_r2}
            r^2\in\R.
        \end{equation}
    Equivalently, $r$ is either purely real or purely imaginary,
        \begin{equation}\label{prop:symm_obstruction_rrealcomplex}
            r\in(\R\cup i\R)\setminus \{0\} \,.
        \end{equation}

    \begin{proof}
        Let $\lambda\in\R$ be the eigenvalue corresponding to $\omegvec$. Then
            \[ \lambda = k_A-lr = k_B-\frac{l}{r}. \]
        Since $k_A,k_B,\lambda\in\R$, it follows that
            \[ lr = k_A-\lambda\in\R \quad \text{ and } \quad \frac{l}{r} = k_B-\lambda\in\R. \]
        As $l\neq 0$ and $r\neq 0$, dividing gives
            \[ r^2 = \frac{lr}{\,l/r\,}\in\R \,. \]
        The equivalence of \Cref{prop:symm_obstruction_r2,prop:symm_obstruction_rrealcomplex} is immediate: if $r^2\in\R$, then $r$ is either real or purely imaginary, and conversely both cases clearly satisfy $r^2\in\R$.
    \end{proof}

\end{proposition}

In \Cref{prop:symmetric_complex_coupling_obstruction}, only the eigenvalue attached to the target eigenvector is assumed real. Since $k_A,k_B\in\R$ make the trace of $\mcsym$ real, the second eigenvalue is then automatically real as well, so this hypothesis is equivalent to requiring a fully real spectrum, as assumed explicitly for the second symmetric model in \Cref{prop:real_coupling_complex_detuning_obstruction} below.

Before turning to the realization statement, it is convenient to separate two remaining spectral parameters once the target ratio $r$ is fixed: the location $\lambda\in\R$ of the distinguished eigenvalue, and the signed offset $\delta \in \R$ locating the second eigenvalue relative to it. For an isolated $2\times2$ sector, the absolute value of one eigenvalue is often not essential, since adding a scalar multiple of the identity shifts both eigenvalues by the same amount and contributes only a global phase to the evolution. However, keeping $\lambda$ explicit is stronger than normalizing it to, for example, $0$: it shows that the target synchronized mode can be placed at any prescribed real spectral location, not merely somewhere on the real axis, and that the second eigenvalue can then be placed at a prescribed signed offset. This is also the natural formulation when the reduced block is viewed as part of a larger operator, where the absolute spectral location may matter.

Now fix $r\in\C^\times$ satisfying $r^2\in\R$. The eigenvector relations
    \[ k_A-lr=\lambda, \quad k_B-\frac{l}{r}=\lambda \]
show that the family of admissible parameters has only one remaining real degree of freedom. Although $l$ a priori carries two real degrees of freedom, requiring $k_A,k_B\in\R$ forces $lr\in\R$, and then
    \[ \frac{l}{r}=\frac{lr}{r^2}\in\R \]
automatically because $r^2\in\R$. It is therefore natural to set
    \[ \tau \coloneqq lr \in \R, \quad\text{equivalently}\quad l=\frac{\tau}{r}, \]
which parametrizes the feasible set by its single intrinsic real parameter. In this parametrization,
    \begin{equation}\label{eqn:ka_kb_dfn_by_tau}
        k_A=\lambda+\tau, \qquad k_B=\lambda+\frac{\tau}{r^2}.
    \end{equation}

Let $\lambda_2$ denote the other eigenvalue of $\mcsym$ \eqref{eqn:mCsym}. Since the trace of any square matrix equals the sum of its eigenvalues (counted with algebraic multiplicity), we have
    \[ \tr(\mcsym)=\lambda+\lambda_2. \]
Therefore,
    \[ \lambda_2=\tr(\mcsym)-\lambda. \]
The trace is then
    \[ \tr(\mcsym) = k_A + k_B = 2\lambda+\tau(1+r^{-2}), \]
hence the second eigenvalue is
    \[ \lambda+\tau(1+r^{-2}). \]
Prescribing a signed offset $\delta$ is therefore equivalent to choosing
    \begin{equation}\label{eqn:sym:spec_gap_delta_dfn}
        \delta=\tau\bigl(1+r^{-2}\bigr).
    \end{equation}

The next proposition shows that this necessary condition is also sufficient: whenever $r^2\in\R$, the complex-symmetric model realizes the target ratio with an arbitrary prescribed real eigenvalue and, except in the degenerate case $r=\pm i$, an arbitrary prescribed signed offset.

\begin{proposition}[Realization for complex symmetric coupling with $r^2 \in \R$] \label{prop:symmetric_complex_coupling_realization}
    Fix $r\in\C^\times$ with $r^2\in\R$. Then there exist $k_A,k_B,\lambda\in\R$ and $l\in\C^\times$ such that
        \begin{equation*}
            \mcsym =
            \begin{pmatrix}
                k_A & -l\\
                -l & k_B
            \end{pmatrix}
        \end{equation*}
    has a real spectrum and satisfies
        \[ \mcsym\omegvec=\lambda\omegvec. \]
    Moreover, if $r\neq\pm i$, then for any prescribed target eigenvalue $\lambda\in\R$ and any prescribed nonzero signed offset $\delta\in\R^\times$, one may choose
    \[
        \tau \coloneqq \frac{\delta}{1+r^{-2}},
        \qquad l \coloneqq \frac{\tau}{r},
        \qquad k_A \coloneqq \lambda+\tau,
        \qquad k_B \coloneqq \lambda+\frac{\tau}{r^2},
    \]
    so that the eigenvalues are $\lambda$ and $\lambda+\delta$.
    \begin{proof}
        Choose any $\lambda\in\R$ and any $\tau\in\R^\times$. Since $r^2\in\R$, the quantities
        \[
            l \coloneqq \frac{\tau}{r},
            \qquad k_A \coloneqq \lambda+\tau,
            \qquad k_B \coloneqq \lambda+\frac{\tau}{r^2}
        \]
        satisfy $k_A,k_B\in\R$ and $l\in\C^\times$. A direct calculation gives
        \[
            k_A-lr = \lambda,
            \qquad k_B-\frac{l}{r} = \lambda,
        \]
        so $\omegvec$ is an eigenvector with eigenvalue $\lambda$.

        The trace equals
            \[ k_A+k_B = 2\lambda+\tau\left(1+r^{-2}\right), \]
        so the second eigenvalue is
            \[ \lambda+\tau\left(1+r^{-2}\right). \]
        If $r\neq\pm i$, then $1+r^{-2}\neq 0$, and choosing
            \[ \tau=\frac{\delta}{1+r^{-2}} \]
        gives second eigenvalue $\lambda+\delta$.
    \end{proof}
\end{proposition}

One can place the complex degree of freedom differently by keeping the coupling real and allowing the detunings to be complex, $l \in \R$ and $k_A,k_B \in \C$. This changes the form of the obstruction, but not the conclusion: the realizable set of quantum states remains a strict subset of $\C^\times$. The next proposition identifies the exact condition in this second symmetric model.

\begin{proposition}[Obstruction for real symmetric coupling with complex detuning] \label{prop:real_coupling_complex_detuning_obstruction}
    Let
    \begin{equation*}
        \mrsym \coloneqq
            \begin{pmatrix}
                k_A & -l\\
                -l & k_B
            \end{pmatrix},
        \quad l \in \R^\times, \quad k_A, k_B\in\C.
    \end{equation*}
    Assume that $\mrsym$ has real spectrum, and let $\omegvec$ be an eigenvector with $\omega_1\omega_2\neq 0$. Define
        \[ r \coloneqq \frac{\omega_2}{\omega_1}\in\C^\times. \]
    Then
        \[ r+\frac{1}{r}\in\R. \]
    Equivalently,
        \[ r\in\R \quad\text{or}\quad |r|=1. \]
    \begin{proof}
        Let $\lambda_1\in\R$ be the eigenvalue corresponding to $\omegvec$. Then
            \[ \lambda_1 = k_A-lr = k_B-\frac{l}{r} \,, \]
        so
            \[ k_A=\lambda_1+lr, \qquad k_B=\lambda_1+\frac{l}{r} \,. \]
        Relating the trace to the eigenvalues of a square matrix in the standard way,
            \[ \tr(\mrsym) = \lambda_1 + \lambda_2 = k_A + k_B \,, \]
        the second eigenvalue is then
            \begin{equation}\label{prop:r_coup_com_det_lambda2_via_trace}
                \lambda_2=(k_A+k_B)-\lambda_1 =\lambda_1+l\left(r+\frac{1}{r}\right).
            \end{equation}
        Since the spectrum is assumed real and $l\in\R^\times$, it follows that
            \[ r+\frac{1}{r}\in\R. \]
        
        Now write $r=\rho e^{i\phi}$ with $\rho>0$. Then
            \[ r+\frac{1}{r} = \left(\rho+\rho\inv\right)\cos\phi + i\left(\rho-\rho\inv\right)\sin\phi. \]
        This is real if and only if
            \[ \left(\rho-\rho\inv\right)\sin\phi=0, \]
        i.e.\ if and only if either $\rho=1$ or $\sin\phi=0$. Equivalently, either $|r|=1$ or $r\in\R$.
    \end{proof}
\end{proposition}

\paragraph{Parameterization.}
Unlike the complex coupling case (\Cref{prop:symmetric_complex_coupling_realization}), where a real auxiliary parameter $\tau$ is introduced to enforce the constraints $k_A,k_B\in\R$, the present setting imposes the realness condition directly on the coupling $l \in \R$. As a result, it is more natural to work directly with $l$. Introducing an auxiliary parameter analogous to $\tau$ would amount to a redundant reparameterization and does not simplify the algebra.

\medskip

The next proposition shows that the obstruction of \Cref{prop:real_coupling_complex_detuning_obstruction}  is also sufficient. Thus, within the real-coupling symmetric model, the exact realizable set of states is
    \[ \left\{r\in\C^\times : r+\frac{1}{r}\in\R\right\}, \]
and, away from the degenerate case $r=\pm i$, the second eigenvalue may again be placed at an arbitrary prescribed signed offset.

\begin{proposition}[Realization for real symmetric coupling with complex detuning] \label{prop:real_coupling_complex_detuning_realization}
    Fix $r\in\C^\times$ such that $r+\frac{1}{r}\in\R.$ Then there exist $l\in\R^\times$, $k_A,k_B\in\C$, and a real eigenvalue $\lambda\in\R$ such that
        \begin{equation*}
            \mrsym = \begin{pmatrix}
                k_A & -l\\
                -l & k_B
            \end{pmatrix}
        \end{equation*}
    has a real spectrum and satisfies
        \[ \mrsym\omegvec=\lambda\omegvec . \]

    Moreover, if $r+\frac{1}{r}\neq 0$ (equivalently $r\neq\pm i$), then for any prescribed target eigenvalue $\lambda\in\R$ and any prescribed nonzero signed offset $\delta \in \R^\times$, one may choose
    \begin{equation*}
        l \coloneqq \frac{\delta}{\,r+\frac{1}{r}\,},
        \qquad k_A \coloneqq \lambda+lr,
        \qquad k_B \coloneqq \lambda+\frac{l}{r},
    \end{equation*}
    so that the eigenvalues are $\lambda$ and $\lambda+\delta$.
    
    \begin{proof}
        Choose any $\lambda \in \R$.
        
        First suppose $r+\frac{1}{r}\neq 0$. Set
        \[
            l \coloneqq \frac{\delta}{r+\frac{1}{r}}\in\R, \qquad
            k_A \coloneqq \lambda+lr, \qquad
            k_B \coloneqq \lambda+\frac{l}{r}.
        \]
        Since $\delta \in \R^\times$ and $r+\frac{1}{r}\in\R^\times$, we have $l\in\R^\times$. Moreover,
            \[ k_A-lr=\lambda, \qquad k_B-\frac{l}{r}=\lambda, \]
        so $\omegvec$ is an eigenvector with eigenvalue $\lambda$, and the second eigenvalue (as in \Cref{prop:r_coup_com_det_lambda2_via_trace}) is
        \[
            (k_A+k_B)-\lambda
            =\lambda+l\left(r+\frac{1}{r}\right)
            =\lambda+\delta\in\R.
        \]
        Hence the spectrum is real.

        If instead $r+\frac{1}{r}=0$ (equivalently $r=\pm i$), choose any $l\in\R^\times$ and define
            \[ k_A\coloneqq \lambda+lr, \qquad k_B\coloneqq \lambda+\frac{l}{r}. \]
        Then $\omegvec$ is again an eigenvector with eigenvalue $\lambda$, and the second eigenvalue also equals $\lambda$, again giving that the spectrum is real. This is a zero-gap realization; it is not covered by the prescribed-offset assertion, since the trace identity forces $\delta=l\left(r+\frac{1}{r}\right)=0$.
    \end{proof}
\end{proposition}

\begin{corollary}[Non-universality of the two symmetric complexifications] \label{cor:nonuniversality_symmetric_models}
    Under a real synchronized-sector spectrum requirement, neither of the two symmetric complexifications above yields arbitrary complex QL-bit states.
    Their exact nonbasis ratio sets are
    \[
        \{r\in\C^\times : r^2\in\R\} = \R^\times\cup i\R^\times %
        \qquad\text{and}\qquad
        \left\{r\in\C^\times : r+\frac{1}{r}\in\R\right\} = \R^\times \cup\{r\in\C^\times:|r|=1\},
    \]
    respectively, whose union is a strict subset of $\C^\times$. In particular, the generic complex ratio $r=2e^{i\pi/4}$ is not realizable by either symmetric model.

    \begin{proof}
        Necessity for the first model follows from \Cref{prop:symmetric_complex_coupling_obstruction}, while sufficiency follows from \Cref{prop:symmetric_complex_coupling_realization}. Likewise, necessity and sufficiency for the second model follow from \Cref{prop:real_coupling_complex_detuning_obstruction,prop:real_coupling_complex_detuning_realization}. The displayed set identities are the equivalent characterizations established in the two obstruction propositions. Their union is a strict subset of $\C^\times$ because, for example, $r=2e^{i\pi/4}$ is neither real nor purely imaginary and does not satisfy $|r|=1$.
    \end{proof}
\end{corollary}

\begin{remark}[Pseudo-Hermitian structure and exceptional points]\label{rem:pseudo_hermitian_ep}
    The non-Hermitian portions of the constrained ratio sets in \Cref{cor:nonuniversality_symmetric_models} are exactly the real-spectrum (``unbroken'') regimes of two canonical pseudo-Hermitian $2\times2$ families \cite{bender1998real,mostafazadeh2002pseudo}. For the real-coupling model, every realization of a unit-modulus target ratio $r=e^{i\phi}$ satisfies $k_B-\lambda=(k_A-\lambda)^*$. Writing $k_{A,B}=a\pm i\gamma$, with $a=\lambda+l\cos\phi$ and $\gamma=l\sin\phi$, exhibits $\mrsym$ as a $\mathcal{PT}$-symmetric dimer with real coupling $l$ and balanced gain and loss $\pm\gamma$:
        \[ \sigma_x\,\overline{\mrsym}\,\sigma_x=\mrsym, \]
    where $\sigma_x$ and $\sigma_z$ denote the usual Pauli matrices and complex conjugation supplies the time-reversal operation \cite{bender1998real,elganainy2018nonhermitian}. The corresponding unbroken condition $l^2\geq\gamma^2$ holds automatically here, since $l^2-\gamma^2=l^2\cos^2\phi$, with equality exactly at $r=\pm i$.
    
    Dually, in the complex-coupling model, the purely imaginary branch $r\in i\R^\times$ forces $l\in i\R^\times$, and the reduced matrix satisfies
        \[ \sigma_z\,\mcsym\,\sigma_z=\mcsym\dagg. \]
    It is therefore pseudo-Hermitian with the indefinite metric $\eta=\sigma_z$ \cite{mostafazadeh2002pseudo}, with the reality boundary again occurring at $r=\pm i$. In both models, the real branch $r\in\R^\times$ reduces to the real-symmetric, hence Hermitian, case.

    At the boundary ratios every nonzero-coupling realization is defective. Indeed, at $r=i$,
    \[
        \mcsym-\lambda I=
        \begin{pmatrix}
            \tau & i\tau\\
            i\tau & -\tau
        \end{pmatrix}
        \quad (\tau\coloneqq lr\in\R^\times),
        \qquad
        \mrsym-\lambda I=
        \begin{pmatrix}
            il & -l\\
            -l & -il
        \end{pmatrix},
    \]
    and each matrix has rank one (its second row is $i$ times its first), so $\lambda$ has algebraic multiplicity two but geometric multiplicity one, and the unique eigendirection is the target $(1,i)^T$ itself; the computation at $r=-i$ is identical with $i\mapsto -i$. These degeneracies are second-order exceptional points \cite{elganainy2018nonhermitian,bergholtz2021exceptional}.
    By contrast, the Hermitian model of \Cref{sec:complex_sym:hermitian} corresponds to the positive metric $\eta=I$ and cannot exhibit exceptional points. More importantly for the synthesis problem, its conjugate off-diagonal structure replaces the phase-sensitive quantities above by the always-real quantity $|r|^2$, so no residual constraint on $r$ survives. The next subsection makes this mechanism explicit.
\end{remark}

\subsection{The Hermitian Remedy}\label{sec:complex_sym:hermitian}
The preceding obstructions arise because ordinary symmetric coupling ties both eigenvalue equations to the same off-diagonal parameter, leaving a residual constraint on the target ratio. We now replace this symmetric placement by Hermitian conjugate pairing. The two coupling directions are then related by $l$ and $l^*$, allowing the target phase to be absorbed into the operator while preserving a real two-level spectrum.

\begin{theorem}[Hermitian coupling realizes arbitrary complex nonbasis states with prescribed real spectral gap] \label{thm:hermitian_arbitrary_state}
    Interpret $k_A,k_B,$ and $l$ as continuously tunable effective parameters. Assume $A=A\dagg$, $B=B\dagg$, and $k_A,k_B\in\R$, and consider the Hermitian block matrix
    \begin{equation}\label{eqn:hermitian_block_matrix}
        H \coloneqq
        \begin{pmatrix}
            A & -C\\[1mm]
            -C\dagg & B
        \end{pmatrix}.
    \end{equation}    
    Assume
        \[
            AV_A = k_A V_A, \qquad
            BV_B = k_B V_B, \qquad
            CV_B = l\,V_A, \qquad
            C\dagg V_A = l^* \,V_B.
        \]
    Let
        \[
            \kpsi = \omegvecvavbspace, \qquad
            |\omega_1|^2+|\omega_2|^2=1, \qquad
            \omega_1\omega_2\neq 0,
        \]
    and define $r \coloneqq \omega_2/\omega_1 \in \C^\times.$ Then for any prescribed eigenvalue $\lambda \in \R$ and any prescribed nonzero signed offset $\delta \in \R^\times$, choosing
        \begin{equation}\label{eqn:hermitian_parameter_choice}
            \tau \coloneqq \frac{\delta}{1+|r|^{-2}}, \qquad
            l \coloneqq \frac{\tau}{r}, \qquad
            k_A \coloneqq \lambda+\tau, \qquad
            k_B \coloneqq \lambda+\frac{\tau}{|r|^2}
        \end{equation}
    makes $\kpsi$ an eigenvector of $H$ and $H|_{\subspace}$ with eigenvalue $\lambda$, and the other eigenvalue of $H|_{\subspace}$ equals $\lambda+\delta$.
    
    (Thus the target state may be chosen as either the lower or upper synchronized eigenstate by taking $\delta>0$ or $\delta<0$, respectively.)

    \begin{proof}
        By \Cref{prop:general_synchronized_reduction} applied with $X=C$ and $Y=C\dagg$, the synchronized subspace
            \[ \subspace=\Span\left\{ \begin{pmatrix}V_A\\0\end{pmatrix}, \begin{pmatrix}0\\V_B\end{pmatrix} \right\} \]
        is $H$-invariant, and
            \[ H|_\subspace= \begin{pmatrix} k_A & -l\\ -l^* & k_B \end{pmatrix}. \]
        Hence $\kpsi$ is an eigenvector with eigenvalue $\lambda$ if and only if
            \[ \lambda = k_A-lr \quad\text{and}\quad \lambda = k_B-\frac{l^*}{r} \,.\]
        With the prescribed choice of \eqref{eqn:hermitian_parameter_choice},
            \[ \tau, k_A, k_B \in \R, \qquad lr = \tau, \qquad \frac{l^*}{r} = \frac{\tau}{|r|^2}, \]
        so indeed
            \[
                k_A-lr = \lambda+\tau-\tau = \lambda, \qquad
                k_B-\frac{l^*}{r} = \lambda+\frac{\tau}{|r|^2}-\frac{\tau}{|r|^2} = \lambda.
            \]
        Thus $\kpsi$ is an eigenvector of $H$ with eigenvalue $\lambda$.
        
        The restriction of $H$ to $\subspace$ (\Cref{prop:general_synchronized_reduction}) is the $2\times 2$ Hermitian matrix
            \[
                \begin{pmatrix}
                    k_A & -l\\[1mm]
                    -l^* & k_B
                \end{pmatrix},
            \]
        so its eigenvalues are real. Since one eigenvalue is $\lambda$, the other is determined by the trace:
            \[ \lambda_2 = (k_A+k_B)-\lambda = \lambda + \tau + \frac{\tau}{|r|^2} = \lambda + \delta. \]
        This proves the claim.
    \end{proof}
\end{theorem}

\begin{remark}[Explicit continuous weighted-operator lifts]\label{rem:continuous_regularity_construction}
    In this continuous setting, $k_A$, $k_B$, and $l$ denote algebraic weighted row sums rather than necessarily nonnegative integer graph degrees, and all three 
    admit explicit full block realizations. For $n,m\geq 2$ and prescribed $k_A,k_B,l\in\C$, take
    \[
        A=\frac{k_A}{n-1}\left(\allonesn\allonesn^T-I_n\right),\qquad
        B=\frac{k_B}{m-1}\left(\allonesm\allonesm^T-I_m\right),\qquad
        C=lV_AV_B\dagg.
    \]
    The diagonal blocks are loopless and complex symmetric, while $C$ is a weighted complete bipartite coupling with entries $l/\sqrt{nm}$. These choices satisfy
        \[ AV_A=k_AV_A, \qquad BV_B=k_BV_B, \qquad CV_B=lV_A, \qquad C^TV_A=lV_B, \qquad C\dagg V_A=l^*V_B. \]
    Thus using $C^T$ or $C\dagg$ as the lower off-diagonal block realizes the complex-symmetric or Hermitian coupling relation, respectively. When $k_A,k_B\in\R$, the diagonal blocks are Hermitian. These weighted complete blocks are convenient witnesses; sparser algebraically regular blocks may be substituted when available.
\end{remark}

The ratio parametrization in \Cref{thm:hermitian_arbitrary_state} ($r=\omega_2/\omega_1$) is convenient for proving existence, but in the Hermitian case one can also write the required effective regularities directly in terms of the target amplitudes. The next lemma records this amplitude-space parameterization and shows that, once the target state and the two spectral parameters $\lambda$ and $\delta$ are fixed, the Hermitian reduced block is uniquely determined.

\begin{lemma}[Direct amplitude-space realization of the Hermitian effective block] \label{lem:hermitian_amplitude_realization}
    With the notation of \Cref{thm:hermitian_arbitrary_state}, write the target synchronized state on the reduced basis of $\subspace$ as
    \[
        \psi \coloneqq \omegvec, \qquad
        |\omega_1|^2+|\omega_2|^2=1, \qquad
        \omega_1\omega_2\neq 0,
    \]
    and define its orthogonal complement direction by
        \[ \psi_\perp \coloneqq \binom{-\omega_2^*}{\omega_1^*}. \]
    For prescribed $\lambda\in\R$ and $\delta\in\R^\times$, consider the Hermitian reduced block (\Cref{prop:general_synchronized_reduction})
    \[
        M_H \coloneqq
        \begin{pmatrix}
            k_A & -l\\[1mm]
            -l^* & k_B
        \end{pmatrix}.
    \]
    Then
        \[ M_H\psi=\lambda\psi, \qquad M_H\psi_\perp=(\lambda+\delta)\psi_\perp \]
    if and only if
    \begin{equation}\label{eqn:herm_kakbl_assignments}
        k_A=\lambda+\delta|\omega_2|^2, \qquad
        k_B=\lambda+\delta|\omega_1|^2, \qquad
        l=\delta\,\omega_1\omega_2^*.
    \end{equation}
    In particular, once the target amplitudes and the spectral gap $(\lambda,\delta)$ are fixed, the effective Hermitian parameters are uniquely determined.

    \begin{proof}
        Because $|\omega_1|^2+|\omega_2|^2=1$, both $\psi$ and $\psi_\perp$ have unit norm, and a direct computation shows $\braket{\psi}{\psi_\perp} = 0$; hence they are orthonormal. Writing out the eigenvector equations for each eigenvector gives
        \begin{equation}\label{eqn:herm_ka_orthog_eig_eqns}
            k_A\omega_1-l\omega_2=\lambda\omega_1, \qquad
            k_A\omega_2^*+l\omega_1^*=(\lambda+\delta)\omega_2^*,
        \end{equation}
        and
        \begin{equation}\label{eqn:herm_kb_orthog_eig_eqns}
            k_B\omega_2 - l^*\omega_1 = \lambda\omega_2, \qquad
            k_B\omega_1^* + l^*\omega_2^* = (\lambda+\delta)\omega_1^*,
        \end{equation}
        respectively. Multiply the first pair \eqref{eqn:herm_ka_orthog_eig_eqns} by $\omega_1^*$ and $\omega_2$, respectively, and add. The $l$-terms cancel, yielding
        \[
            k_A\bigl(|\omega_1|^2+|\omega_2|^2\bigr)
            = \lambda|\omega_1|^2+(\lambda+\delta)|\omega_2|^2,
        \]
        hence
            \[ k_A=\lambda+\delta|\omega_2|^2. \]
        Similarly, multiplying the second pair \eqref{eqn:herm_kb_orthog_eig_eqns} by $\omega_2^*$ and $\omega_1$ and adding eliminates $l^*$ and gives
            \[ k_B=\lambda+\delta|\omega_1|^2. \]
        Substituting this into the first eigenvector equation gives
        \[
            l=\frac{(k_A-\lambda)\omega_1}{\omega_2}
            = \delta|\omega_2|^2\frac{\omega_1}{\omega_2}
            = \delta\,\omega_1\omega_2^*.
        \]
        This proves necessity.
        
        Conversely, with these choices,
        \[
            k_A\omega_1-l\omega_2
            = (\lambda+\delta|\omega_2|^2)\omega_1 - \delta\,\omega_1\omega_2^*\,\omega_2
            = \lambda\omega_1,
        \]
        and
        \[
            -l^*\omega_1+k_B\omega_2
            = -\delta\,\omega_1^*\omega_2\,\omega_1 + (\lambda+\delta|\omega_1|^2)\omega_2
            = \lambda\omega_2,
        \]
        so $M_H\psi=\lambda\psi$. Also,
            \[ \tr(M_H)=k_A+k_B=2\lambda+\delta. \]
        Hence the other eigenvalue is $\lambda+\delta$ (as the trace also equals the sum of the eigenvalues). Since $M_H$ is Hermitian and $\delta\neq 0$, the corresponding eigenspace is orthogonal to $\psi$, namely $\Span\{\psi_\perp\}$. Therefore
            \[ M_H\psi_\perp=(\lambda+\delta)\psi_\perp. \]
        Uniqueness is immediate from the necessity argument.
    \end{proof}
\end{lemma}

Using $r=\omega_2/\omega_1$, this is equivalent to the parameter choice \eqref{eqn:hermitian_parameter_choice}. More importantly, it already reveals the basic Hermitian mechanism: the diagonal terms depend only on the population weights $|\omega_1|^2$ and $|\omega_2|^2$, while the off-diagonal coupling stores the relative phase through the product $\omega_1\omega_2^*$. We now explain why that mechanism removes the obstructions present in the two symmetric models.

\paragraph{Why the Hermitian construction works.}
The two symmetric models fail for the same structural reason: the same scalar $l$ appears in both off-diagonal positions, so the reduced eigenvalue equations involve the pair $lr$ and $l/r$ without complex conjugation. As a result, the phase of $r$ survives in the consistency conditions. In the complex-symmetric model with $k_A,k_B,\lambda\in\R$, requiring
    \[ \lambda = k_A-lr, \qquad \lambda = k_B-\frac{l}{r} \]
forces both $lr$ and $l/r$ to be real, hence $r^2\in\R$. In the real-coupling symmetric model with $l\in\R$, the second eigenvalue is
    \[ \lambda_2=\lambda+l\left(r+\frac{1}{r}\right), \]
so the real spectrum forces
    \[ r+\frac{1}{r}\in\R. \]
Thus, in both symmetric placements of the complex degree of freedom, the target phase encoded in $r$ is not fully absorbed by the matrix structure; instead it reappears as an additional reality constraint on $r$.

The Hermitian reduction
    \[
    M_H=\begin{pmatrix}
            k_A & -l\\
            -l^* & k_B
        \end{pmatrix}
    \]
changes exactly this feature. Its eigenvalue equations are
    \[ \lambda = k_A-lr, \qquad \lambda = k_B-\frac{l^*}{r}. \]
If one chooses
    \[ l=\frac{\tau}{r},\qquad \tau\in\R, \]
then
    \[ lr=\tau, \qquad \frac{l^*}{r}=\frac{\tau}{|r|^2}. \]
The phase of $r$ is therefore absorbed into the conjugate pair $l$ and $l^*$, and the consistency conditions depend only on $|r|^2$, not on $r^2$ or on $r+1/r$. Since $|r|^2\in\R_{>0}$ for every $r\in\C^\times$, no residual phase obstruction remains. Moreover, the spectral gap set to
    \[ \delta=\tau\left(1+|r|^{-2}\right) \]
is solvable for every prescribed $\delta\in\R^\times$, because
    \[ 1+|r|^{-2}>0. \]
At the scalar level, this is the key change: Hermitian symmetry replaces the problematic quantities $r^2$ and $r+1/r$ by the always-real quantity $|r|^2$.

Equivalently, on the synchronized two-level space, Hermitian symmetry allows the target state and its orthogonal complement to serve as an eigenbasis. The relative phase is then stored in a conjugate pair of off-diagonal entries, while the effective eigenvalues remain real by construction.

\smallskip

The preceding universality statements are formulated for nonbasis states because the ratio $r=\omega_2/\omega_1$ is only defined when both amplitudes are nonzero. The omitted endpoints are not difficult graph-theoretic cases; they are boundary points of the same two-level state space where the coupling between the two synchronized modes must vanish. We record this separately before turning from the reduced $2\times 2$ block to the full graph-supported operator.

\begin{remark}[Basis states]\label{rem:basis_states}
    The singular cases $\omega_1=0$ or $\omega_2=0$ correspond to basis states and lie outside the ratio parametrization $r=\omega_2/\omega_1$. In both the symmetric and Hermitian coupled models, these states occur exactly in the decoupled case $l=0$; with nonzero coupling they are obtained only as limits $|r|\to 0$ or $|r|\to\infty$.
\end{remark}

\paragraph{Further structural consequences.}
\Cref{thm:hermitian_arbitrary_state} concerns the synchronized restriction $H|_\subspace$. The next proposition and corollary show that, in the Hermitian setting, this restriction is not merely a computational device: it defines an exact orthogonally decoupled two-level subsystem of the full network. The proofs are deferred to \Cref{app:herm_structural_consequences_proofs}.

\begin{proposition}[The Hermitian construction defines an exact closed two-level subsystem] \label{prop:hermitian_closed_subsystem}
    Under the hypotheses of \Cref{thm:hermitian_arbitrary_state}, let
    \[
        \subspace:=\operatorname{span}\left\{
        \begin{pmatrix}V_A\\0\end{pmatrix},
        \begin{pmatrix}0\\V_B\end{pmatrix}
        \right\}.
    \]
    Then both $\subspace$ and its orthogonal complement
    \[
        \subspace^\perp =
        \left\{
        \begin{pmatrix}x\\y\end{pmatrix} :
        \braket{V_A}{x}=0,\ \braket{V_B}{y}=0
        \right\}
    \]
    are invariant under $H$. Equivalently, relative to the orthogonal decomposition
        \[ \C^{\,n+m}=\subspace\oplus \subspace^\perp, \]
    the Hermitian matrix $H$ is block diagonal:
        \[ H = H|_\subspace \oplus H|_{\subspace^\perp}. \]
    (In operator-theoretic language, $\subspace$ is a reducing subspace for $H$.)
\end{proposition}

\begin{corollary}[The designed synchronized eigenpair is unique in the full operator] \label{cor:hermitian_unique_eigenpair}
    Under the hypotheses of \Cref{thm:hermitian_arbitrary_state}, suppose in addition that
        \[ \{\lambda,\lambda+\delta\}\cap \sigma(H|_{\subspace^\perp})=\varnothing, \]
    where $\sigma(\cdot)$ denotes the spectrum. Then $\lambda$ and $\lambda+\delta$ are simple eigenvalues of the full matrix $H$, and every eigenvector of $H$ with either of these eigenvalues lies entirely in $\subspace$.
\end{corollary}

\begin{remark}[An explicit continuous isolation family] \label{rem:continuous_full_isolation}
    For the loopless weighted blocks
        \[ A=\frac{k_A}{n-1} \left(\allonesn\allonesn^T-I_n\right), \qquad B=\frac{k_B}{m-1} \left(\allonesm\allonesm^T-I_m\right), \]
    together with $C=lV_AV_B\dagg$, the coupling vanishes on $\subspace^\perp$ and
        \[ H|_{\subspace^\perp} = -\frac{k_A}{n-1}I_{n-1} \oplus -\frac{k_B}{m-1}I_{m-1}. \]
    Thus the only possible complementary collisions with the designed eigenvalues are
        \[ -\frac{k_A}{n-1}\in\{\lambda,\lambda+\delta\}, \qquad -\frac{k_B}{m-1}\in\{\lambda,\lambda+\delta\}. \]
    For fixed target data, these equations exclude only finitely many choices of $n$ and $m$. Consequently, infinitely many block sizes yield simple designed eigenvalues in the full continuous weighted operator.
\end{remark}

\paragraph{Interpretation and practical implications.} \Cref{prop:hermitian_closed_subsystem} clarifies the status of the synchronized reduction. The two-component form $\omegvecvavbT$ is not merely a trial restriction to block-constant vectors. In the Hermitian model, $\subspace$ is a reducing subspace: both $\subspace$ and $\subspace^\perp$ are invariant. Equivalently, after choosing an orthonormal basis adapted to
    \[ \C^{\,n+m}=\subspace\oplus \subspace^\perp, \]
the full matrix $H$ becomes block diagonal, so the unitary time evolution splits as
    \[ e^{-itH}=e^{-itH|_\subspace}\oplus e^{-itH|_{\subspace^\perp}}. \]
Thus an initial state in the synchronized sector never leaks into nonsynchronized modes, and nonsynchronized modes never feed amplitude back into the synchronized sector. In this precise sense, the Hermitian construction embeds an exact QL-bit inside the larger graph, rather than providing only an approximate two-level description.

The corollary gives the corresponding spectral interpretation. If $H|_{\subspace^\perp}$ does not contain either designed eigenvalue $\lambda$ or $\lambda+\delta$, then the two synchronized eigenstates are uniquely identifiable in the full operator. This matters for numerical extraction, spectral addressing, and state preparation: the designed QL-bit eigenvectors cannot be confused with nonsynchronized eigenvectors at the same eigenvalue. This also means that one may enlarge or modify the rest of the network while preserving the intended synchronized QL-bit, provided no spectral collision with $\lambda$ or $\lambda+\delta$ is introduced in the complementary sector.
For example, a Hermitian perturbation
\[
    \Delta H=
    \begin{pmatrix}
        \Delta A & -\Delta C\\
        -\Delta C\dagg & \Delta B
    \end{pmatrix}
\]
leaves the synchronized block unchanged if
    \[ \Delta A V_A=0,\qquad \Delta B V_B=0,\qquad \Delta C V_B=0,\qquad \Delta C\dagg V_A=0. \]
These conditions say that the perturbation has zero action on the synchronized modes, or equivalently that it changes only the nonsynchronized degrees of freedom. Designing $H|_{\subspace^\perp}$ is therefore a separate structured spectral problem: local changes in $A$, $B$, or $C$ can move many complementary eigenvalues, and avoiding collisions with $\lambda$ and $\lambda+\delta$ generally requires either explicit spectral estimates or direct spectral verification for the chosen finite graph. \Cref{sec:discrete_and_density:worked_example} gives a minimal balanced instance in which both outcomes occur: an exact collision produced by phase cancellation on a complementary mode, and its removal by a different admissible choice of coupling matchings.

\section{Complex Edge Weights in Asymmetric Coupling and Detuning}\label{sec:complex_asym}
The Hermitian construction restores universality by pairing the off-diagonal entries by conjugation. A second route is to drop the symmetry requirement altogether and allow the two coupling directions to be tuned independently. We first isolate the genuinely asymmetric mechanism by keeping a common synchronized regularity on both sides, i.e. by imposing $k_A=k_B=k$ while allowing the two coupling blocks to differ. This shows that universality can already be recovered from directed coupling alone, without any diagonal detuning. We then pass to the fully generalized setting (\Cref{sec:generalized_complex}) in which the diagonal regularities are also allowed to differ, so that asymmetry and detuning are both present.

On the synchronized subspace, we replace the single effective coupling $l$ of \Cref{sec:prelims:synchronized_reduction} by two independent scalars $l_A$ and $l_B$ in the reduced $2\times 2$ block.

Throughout \Cref{sec:complex_asym,sec:generalized_complex}, $\lambda$ denotes the eigenvalue assigned to the target state, and the other synchronized eigenvalue is written as $\lambda+\delta$, where $\delta\in\R^{\times}$ is a signed offset. Thus $\delta>0$ makes the target state the lower synchronized eigenstate, while $\delta<0$ makes it the upper synchronized eigenstate. In the common-regularity parametrization of \Cref{thm:asymmetric_common_k}, $\tau=\delta/2$ and $k=\lambda+\tau$, so the two synchronized eigenvalues are
    \[ k-\tau=\lambda,\qquad k+\tau=\lambda+\delta . \]

\begin{theorem}[Asymmetric coupling with common regularity realizes arbitrary complex nonbasis states with prescribed real spectral gap]
\label{thm:asymmetric_common_k}
    Interpret $k$, $l_A$, and $l_B$ as continuously tunable effective parameters, and consider the block matrix
    \[
        \rasym \coloneqq
        \begin{pmatrix}
            A & -C_{A}\\[1mm]
            -C_{B} & B
        \end{pmatrix}.
    \]
    Assume
    \[
        AV_A = kV_A, \qquad
        BV_B = kV_B, \qquad
        C_{A}V_B = l_A V_A, \qquad
        C_{B}V_A = l_B V_B.
    \]
    Let
    \[
        \kpsi =
        \begin{pmatrix}
            \omega_1 V_A\\[1mm]
            \omega_2 V_B
        \end{pmatrix}, \qquad
        |\omega_1|^2+|\omega_2|^2=1, \qquad
        \omega_1\omega_2\neq 0,
    \]
    and define $r\coloneqq \omega_2/\omega_1\in\C^\times$. Then for any prescribed eigenvalue $\lambda\in\R$ and any prescribed nonzero signed offset $\delta\in\R^\times$, choosing
    \[
        \tau \coloneqq \frac{\delta}{2}, \qquad
        k \coloneqq \lambda+\tau, \qquad
        l_A \coloneqq \frac{\tau}{r}, \qquad
        l_B \coloneqq \tau r
    \]
    makes $\kpsi$ an eigenvector of $\rasym$ and $\rasym|_{\subspace}$ (\Cref{sec:prelims:synchronized_reduction}) with eigenvalue $\lambda$, and the other eigenvalue of $\rasym|_{\subspace}$ equals $\lambda+\delta$.

    \begin{proof}
        By \Cref{prop:general_synchronized_reduction} applied with $X=C_A$ and $Y=C_B$, the synchronized subspace
            \[ \subspace=\Span\left\{ \begin{pmatrix}V_A\\0\end{pmatrix}, \begin{pmatrix}0\\V_B\end{pmatrix} \right\} \]
        is $\rasym$-invariant, and
            \[ \rasym|_\subspace= \begin{pmatrix} k & -l_A\\ -l_B & k \end{pmatrix}. \]
        Hence $\kpsi$ is an eigenvector with eigenvalue $\lambda$ if and only if
            \[ \lambda = k - l_A r, \quad \text{and} \quad \lambda = k - \frac{l_B}{r}. \]
        With the prescribed choice,
        \[
            l_A r=\tau, \qquad
            \frac{l_B}{r}=\tau,
        \]
        so indeed
        \[
            k-l_A r = \lambda+\tau-\tau = \lambda, \qquad
            k-\frac{l_B}{r} = \lambda+\tau-\tau = \lambda.
        \]
        Thus $\kpsi$ is an eigenvector of $\rasym$ with eigenvalue $\lambda$.

        The trace of $\rasym|_\subspace$ is $2k=2\lambda+2\tau=2\lambda+\delta$. Since one eigenvalue is $\lambda$, the other is
        \[
            2k-\lambda
            = \lambda+\delta.
        \]
        This proves the claim.
    \end{proof}
\end{theorem}

Thus, imposing $k_A=k_B$ does not destroy universality so long as the common value $k$ remains tunable: the asymmetric couplings $l_A$ and $l_B$ encode the target ratio $r$, while the common diagonal level $k$ sets the spectral center.

\begin{remark}
    Similarly to \Cref{rem:continuous_regularity_construction}, in the asymmetric case, arbitrary effective couplings may be realized continuously by
        \[ C_A=l_A V_A V_B\dagg,\qquad C_B=l_B V_B V_A\dagg . \]
    
\end{remark}

\section{Generalized Construction and Unification Theorem}\label{sec:generalized_complex}

The preceding theorem shows that the essential gain from asymmetry already appears in the common-regularity case $k_A = k_B = k$: once the two coupling directions are allowed to differ, the phase obstruction present in the symmetric models disappears, because $l_A$ and $l_B$ can encode the target ratio $r \in \C^\times$ independently, while the common diagonal level $k$ fixes the spectral center of the synchronized $2 \times 2$ block. The generalized setting considered next does not enlarge the realizable set of states further; rather, it adds independent diagonal detuning $k_A \neq k_B$ on top of the asymmetric couplings, thereby separating more cleanly the state-synthesis role of the off-diagonal parameters from the spectral-placement role of the diagonal terms. We therefore now pass from the purely asymmetric common-regularity model to the fully generalized asymmetric-detuned reduction, in which both coupling directions and both diagonal levels are treated as independent effective parameters. We proceed using the subspace reduction derived in \Cref{prop:general_synchronized_reduction}.

\begin{theorem}[Generalized coupling realizes arbitrary complex nonbasis states with prescribed real spectral gap]\label{thm:generalized_asymmetric}
    Let
        \[
            \kpsi= \omegvecvavbspace, \qquad
            |\omega_1|^2+|\omega_2|^2=1, \qquad
            \omega_1\omega_2\neq 0,
        \]
    and define
        \[ r\coloneqq \frac{\omega_2}{\omega_1}\in\C^\times. \]
    Then for any prescribed eigenvalue $\lambda\in\R$ and any prescribed nonzero signed offset $\delta\in\R^\times$, one may choose any $\tau_A,\tau_B\in\R$ satisfying
        \[ \tau_A+\tau_B=\delta, \]
    and set
        \[
            l_A \coloneqq \frac{\tau_A}{r}, \qquad
            l_B \coloneqq \tau_B r, \qquad
            k_A \coloneqq \lambda+\tau_A, \qquad
            k_B \coloneqq \lambda+\tau_B.
        \]
    Assume that the blocks of $R_{\mathrm{asym}}$ realize these chosen parameters, i.e.,
        \[ AV_A=k_AV_A, \qquad BV_B=k_BV_B, \qquad C_AV_B=l_AV_A, \qquad C_BV_A=l_BV_B. \]
    Then $\kpsi$ is an eigenvector of $\rasym$ with eigenvalue $\lambda$, and the other eigenvalue of $\rasym|_\subspace$ equals $\lambda+\delta$.

    \begin{proof}
        By \Cref{prop:general_synchronized_reduction}, the synchronized subspace
        $\subspace$ is invariant under $R_{\mathrm{asym}}$, and
            \[ R_{\mathrm{asym}}|_{\subspace} = \begin{pmatrix} k_A & -l_A\\ -l_B & k_B \end{pmatrix}. \]
        A direct calculation gives
            \[ k_A-l_A r=\lambda, \qquad k_B-\frac{l_B}{r}=\lambda, \]
        so $\kpsi$ is an eigenvector with eigenvalue $\lambda$. The trace of the effective block is
            \[ k_A+k_B = 2\lambda+\tau_A+\tau_B = 2\lambda+\delta, \]
        so the second eigenvalue is
            \[ (k_A+k_B)-\lambda = \lambda+\delta. \]
    \end{proof}
\end{theorem}
The obstruction disappears here through a mechanism similar to but distinct from the Hermitian one. Hermitian coupling relates the two directions by complex conjugation, so the consistency conditions depend on $|r|^2$. In the asymmetric model, the two coupling directions are independent, and
    \[ l_A r=\tau_A, \qquad \frac{l_B}{r}=\tau_B \]
can instead be chosen separately. The target phase of $r$ therefore imposes no algebraic restriction. The tradeoff is that the effective block need not be Hermitian or even normal, so one gains universality without the orthogonal decomposition and spectral guarantees that accompany Hermitian symmetry.

\medskip
We collect the individual results of \Cref{sec:complex_sym:why_sym_not_work,sec:complex_sym:hermitian,sec:complex_asym} in the following theorem.

\begin{theorem}[Universality taxonomy of effective QL-bit reductions]
    Fix nonbasis target coefficients $\omega_1,\omega_2$ with $|\omega_1|^2+|\omega_2|^2=1$, and define $r \coloneqq \omega_2/\omega_1 \in \C^\times$. Let $\lambda\in\R$ and $\delta\in\R^\times$ be prescribed.

    For the reduced $2\times 2$ models:
    \begin{enumerate}[(i)]
        \item complex-symmetric coupling with $k_A,k_B\in\R$, $l\in\C^\times$,
        \item real coupling with $l\in\R^\times$, $k_A,k_B\in\C$,
        \item Hermitian coupling,
        \item asymmetric coupling, %
    \end{enumerate}
    the corresponding sets of realizable nonbasis ratios with prescribed real eigenvalue $\lambda$ and prescribed signed offset $\delta$ are, respectively,
    \begin{enumerate}[(i)]
        \item $\{r\in\C^\times:r^2\in\R,\ r\neq\pm i\} = (\R^\times\cup i\R^\times)\setminus\{\pm i\}$ (\Cref{prop:symmetric_complex_coupling_realization}),
        \item $\left\{r\in\C^\times:r+\frac{1}{r}\in\R,\ r\neq\pm i\right\} = \left(\R^\times\cup\{r\in\C^\times:|r|=1\}\right)\setminus\{\pm i\}$ (\Cref{prop:real_coupling_complex_detuning_realization}),
        \item $\C^\times$ (\Cref{thm:hermitian_arbitrary_state}),
        \item $\C^\times$ (\Cref{thm:asymmetric_common_k}).
    \end{enumerate}
    For the nonzero couplings assumed in the first two models, the exceptional points $r=\pm i$ admit only defective zero-gap realizations: the reduced matrix has a repeated eigenvalue with geometric multiplicity one, and the unique eigendirection is the target itself (\Cref{rem:pseudo_hermitian_ep}).

    \begin{proof}
        In cases (i) and (ii), necessity follows from \Cref{prop:symmetric_complex_coupling_obstruction,prop:real_coupling_complex_detuning_obstruction}, while sufficiency away from $r=\pm i$ follows from \Cref{prop:symmetric_complex_coupling_realization,prop:real_coupling_complex_detuning_realization}. Because $\delta\neq 0$, the ratios $r=\pm i$ are excluded: at those points both symmetric models force the synchronized eigenvalue offset to vanish. Their defectiveness is established in \Cref{rem:pseudo_hermitian_ep}. Cases (iii) and (iv) are \Cref{thm:hermitian_arbitrary_state,thm:asymmetric_common_k}, respectively.
    \end{proof}
\end{theorem}

\subsection{Magic State Realizability}
The universality taxonomy immediately recovers the status of the canonical one-qubit $H$- and $T$-type examples.

\begin{corollary}[Canonical $H$- and $T$-state examples]
\label{cor:canonical_magic_state_examples}
Fix a prescribed target eigenvalue $\lambda\in\R$ and a prescribed nonzero signed offset $\delta\in\R^\times$. Consider the reduced computational-basis states
\[
    \ket{H} \coloneqq \cos\!\left(\frac{\pi}{8}\right)\ket{0}+\sin\!\left(\frac{\pi}{8}\right)\ket{1},
    \qquad
    \ket{T} \coloneqq \frac{1}{\sqrt{2}}\bigl(\ket{0}+e^{i\pi/4}\ket{1}\bigr).
\]
Then:
\begin{enumerate}
    \item $\ket{H}$ is realizable in each of the four effective models of \Cref{sec:complex_sym,sec:complex_asym,sec:generalized_complex};
    \item $\ket{T}$ is realizable in the real-coupling symmetric model, the Hermitian model, and the asymmetric model, but not in the complex-symmetric model with real diagonal detuning.
\end{enumerate}
    \begin{proof}
        For $\ket{H}$, the amplitude ratio is
        \[
            r_H \coloneqq \frac{\sin(\pi/8)}{\cos(\pi/8)}
            = \tan\!\left(\frac{\pi}{8}\right)
            = \sqrt{2}-1 \in \R.
        \]
        Hence $r_H^2\in\R$ and $r_H + \frac{1}{r_H}\in\R$, with $r_H \neq \pm i$. Therefore $\ket{H}$ is realizable in the two symmetric models by \Cref{prop:symmetric_complex_coupling_realization,prop:real_coupling_complex_detuning_realization}, and in the Hermitian and asymmetric models by \Cref{thm:hermitian_arbitrary_state,thm:asymmetric_common_k}.

        For $\ket{T}$, the amplitude ratio is
        \[
            r_T \coloneqq e^{i\pi/4}.
        \]
        Then
        \[
            r_T^2 = i \notin \R,
            \qquad
            r_T+\frac{1}{r_T}
            = e^{i\pi/4}+e^{-i\pi/4}
            = 2\cos\frac{\pi}{4}
            = \sqrt{2}\in\R,
        \]
        and $r_T\neq\pm i$. Hence \Cref{prop:symmetric_complex_coupling_realization} excludes $\ket{T}$ from the complex-symmetric model with real diagonal detuning, while \Cref{prop:real_coupling_complex_detuning_realization} realizes it in the real-coupling symmetric model. The Hermitian and asymmetric realizations again follow from \Cref{thm:hermitian_arbitrary_state,thm:asymmetric_common_k}.
    \end{proof}
\end{corollary}

\section{Discrete Realizations and Density of Representation}\label{sec:discrete_and_density}
\Cref{sec:complex_sym,sec:complex_asym,sec:generalized_complex} are effective-parameter results: the reduced quantities $k_A,k_B,l$ (and, in the asymmetric setting, $l_A,l_B$) are treated as continuously tunable scalars. For actual finite weighted graphs, however, these parameters are discrete.
The goal of this section is therefore to identify the exact arithmetic row- and column-sum constraints within a balanced model and then show that, as the graph size grows, the resulting exact finite realizations become dense in the synchronized pure-state space.

\subsection{Admissible Discrete Coupling Phases}
\begin{definition}[Admissible coupling phases]\label{sec:complex_sym:coupling_phases}
    Define the fourth roots of unity as
        \begin{equation}\label{eqn:fourth_roots}
            \mu_4 \coloneqq \{\pm 1, \pm i\} = \left\{ e^{i\pi m/2}  \given m \in \{0,1,2,3\} \right\}.
        \end{equation}
    We write $\mu_4^0 \coloneqq \{0\} \cup \mu_4$ for the set of admissible weighted adjacency entries.
\end{definition}

We now choose a concrete finite alphabet of allowed coupling weights. The set $\mu_4^0$ should be read as a unit-weight four-phase discretization: $0$ denotes an absent edge, while $1$, $-1$, $i$, and $-i$ denote present edges with phases $0$, $\pi$, $\pi/2$, and $3\pi/2$. This alphabet is small but structurally natural for the Hermitian setting: it is closed under complex conjugation and supplies positive and negative real and imaginary contributions separately. Consequently, row and column sums of such coupling blocks lie in the Gaussian integers $\Z[i]$.
Weighted graphs whose edges carry unit-modulus phases are studied in spectral graph theory as \emph{complex unit gain graphs}, with signed graphs as the two-phase special case \cite{zaslavsky1982signed,reff2012spectral}; in physics, the same objects appear as magnetic hopping operators whose edges carry $U(1)$ Peierls phases \cite{lieb1993fluxes}. In this language, $\mufz$-weighted couplings are gain graphs over the cyclic group $\{\pm 1,\pm i\}\cong\Z_4\subset U(1)$, and the Hermitian model of \Cref{sec:complex_sym:hermitian} is precisely the magnetic case, in which reciprocal edges carry conjugate phases.

The fourth-root alphabet is not the only discrete phase alphabet that could be studied. Among the full root-of-unity alphabets $\mu_p$, however, $\mu_4$ is the smallest that contains a nonreal phase and is closed under both complex conjugation and negation, making it the canonical minimal choice for the signed real-imaginary coordinate construction used here. The real alphabet $\{0,\pm1\}$ is too small, since it forces every coupling row-sum to be real and therefore cannot encode a generic Hermitian coherence. Odd root alphabets such as $\mu_5$ are conjugation-closed but not centrally symmetric: if $\zeta\in\mu_5$, then $-\zeta\notin\mu_5$, so signed cancellation is not available edge-by-edge and row-sum admissibility is governed by cyclotomic relations rather than by independent positive and negative counts. Even alternatives such as $\mu_6$ are possible in principle, but they produce an Eisenstein-integer triangular lattice rather than the Cartesian Gaussian-integer lattice; this would give a different arithmetic theory without strengthening the density result. Larger alphabets similarly add cyclotomic arithmetic and extra exact finite cases, but are unnecessary for obtaining a dense set of exact realizations.

With $\mu_4^0$, every effective coupling has the form $l=c+di\in\Z[i]$. Any row with this sum must contain at least $|c|+|d|$ nonzero entries, because the real and imaginary components require at least $|c|$ and $|d|$ signed unit contributions, respectively. The perfect-matching construction in \Cref{prop:balanced_discrete_coupling} attains this lower bound simultaneously in every row and column. Thus a balanced $q\times q$ coupling realizes $l$ exactly when $|c|+|d|\leq q$. In this precise sense, $\mu_4^0$ is the minimal canonical alphabet for the Gaussian-integer construction developed here: it is closed under Hermitian conjugation, gives independent signed real and imaginary coupling control, yields a simple exact admissibility condition, and produces a Gaussian-rational dense set of target ratios.

\subsection{Balanced Discrete Coupling}
We work throughout in the balanced case $n=m=q$, with $q$ even. This parity assumption removes the usual obstruction for ordinary simple regular graphs: when $q$ is even, $qk$ is even for every $k\in\{0,1,\ldots,q-1\}$, so $q$-vertex $k$-regular diagonal blocks exist for all admissible degrees.
Thus, the only genuinely nontrivial discrete issue is the coupling block: the diagonal degrees can be realized by standard regular graphs, but achieving a prescribed complex effective coupling requires a row- and column-sum construction with entries restricted to $\mu_4^0$.

Necessarily, if a balanced $\mu_4^0$-valued coupling block $C$ has common weighted row and column sum $l=c+di$, then $l\in\Z[i]$ and $|c|+|d|\le q$, since each row sum is a sum of $q$ terms from $\mufz$. The next proposition shows that this necessary condition is also sufficient: it characterizes the discrete set of balanced effective couplings $l$ (and hence conjugate couplings $l^*$) available to the Hermitian QL-bit model, and it gives an explicit realization via disjoint perfect matchings.

\begin{proposition}[Balanced discrete coupling realization via perfect matchings] \label{prop:balanced_discrete_coupling}
    Fix $q\in\N$, and let
        \begin{equation}
            \mathcal{L}_q \coloneqq \{\,c+di\in\Z[i]: |c|+|d|\le q\,\}.
        \end{equation}
    If $l=c+di\in\mathcal{L}_q$, then there exists a matrix $C\in(\mu_4^0)^{q\times q}$ such that
        \[ C\allonesq=l\,\allonesq,\qquad C^T\allonesq=l\,\allonesq. \]
    Equivalently, every $l \in \mathcal{L}_q$ is realizable as a balanced effective coupling.

    \begin{proof}
        Let $P_0,\dots,P_{q-1}$ be permutation matrices arising from a $1$-factorization of the complete bipartite graph $K_{q,q}$ on two parts of size $q$; for example, one may take the cyclic shift matrices on $\Z_q$. Thus the $P_j$ are pairwise disjoint, and each row and each column of each $P_j$ contains exactly one $1$.

        Define
            \[
                C \coloneqq \sign(c)\sum_{j=0}^{|c|-1}P_j
                       + i\,\sign(d)\sum_{j=|c|}^{|c|+|d|-1}P_j \,,
            \]
        where $\sign(0)\coloneqq 0$, and empty sums are understood as zero. Equivalently, assign $|c|$ values of each row and column to $1$, and then $|d|$ values of each row and column to $i$ (modulo a sign). Since $|c|+|d|\le q$, the chosen matchings are disjoint, so every entry of $C$ lies in $\mufz=\mu_4^0$. Moreover, each row and each column meets exactly $|c|$ terms in the first block and exactly $|d|$ terms in the second, hence every row-sum and every column-sum equals
            \[ \sign(c)\,|c| + i\,\sign(d)\,|d| = c+di=l. \]
        Therefore $C\allonesq=l\,\allonesq$ and $C^T\allonesq=l\,\allonesq$, as claimed.
    \end{proof}
\end{proposition}

The $1$-factorization in the proof is not unique: any assignment of $|c|+|d|$ pairwise disjoint perfect matchings to the four signed phase directions realizes the same balanced coupling $l$, while different assignments generally produce different spectra on the complementary subspace $\subspace^\perp$. \Cref{sec:discrete_and_density:worked_example} exploits exactly this freedom.

\smallskip

\noindent\emph{Hermitian implication.}
With $V_A=V_B=\onesqrtq\allonesq$, any such coupling block satisfies
    \[ CV_B=l\,V_A,\qquad C\dagg V_A=l^*V_B, \]
since
    \[ C\dagg V_A = \onesqrtq\,\overline{C^T}\allonesq = \onesqrtq\,l^* \allonesq = l^* V_B. \]
Hence every $l\in\mathcal{L}_q$ yields an exact balanced discrete Hermitian off-diagonal block, and after choosing $A$ and $B$ so that $AV_A=k_A V_A$ and $BV_B=k_B V_B$, one obtains a discrete realization of the Hermitian effective matrix
    \[
        \begin{pmatrix}
            k_A & -l\\
            -l^* & k_B
        \end{pmatrix}.
    \]

Thus \Cref{prop:balanced_discrete_coupling} resolves the coupling-level arithmetic problem. The next corollary is the descent step: it combines this coupling realization with \Cref{thm:hermitian_arbitrary_state} and the existence of regular graphs, showing that admissible effective parameters produce an actual finite Hermitian QL-bit.

\subsection{From Arithmetic Admissibility to Density of Discrete Representation}
We continue to use the ratio parameter
    \[ r \coloneqq \frac{\omega_2}{\omega_1}\in\C^\times, \qquad |\omega_1|^2 + |\omega_2|^2 = 1, \qquad \omega_1\omega_2 \neq 0. \]
Since $\Q(i)$ is dense in $\C$, density will follow once we show that every Gaussian-rational ratio admits an exact balanced discrete Hermitian realization. The next corollary gives the required descent statement from the continuous Hermitian model to the discrete setting; we then verify that every Gaussian-rational ratio satisfies it.

\begin{corollary}[Balanced discrete Hermitian realization under arithmetic admissibility] \label{cor:balanced_discrete_hermitian}
    Fix $r \in\C^\times$, $\lambda\in\R$, and $\delta\in\R^\times$, and define %
        \[
            \tau \coloneqq \frac{\delta}{1+|r|^{-2}}, \qquad
            l \coloneqq \frac{\tau}{r}, \qquad
            k_A \coloneqq \lambda+\tau, \qquad
            k_B \coloneqq \lambda+\frac{\tau}{|r|^2}
        \]
    as in \Cref{thm:hermitian_arbitrary_state}. Suppose there exists an even integer $q$ such that
        \[ l\in\mathcal{L}_q, \qquad k_A,k_B\in\{0,1,\dots,q-1\}. \]
    Then there exist simple $q$-vertex regular graphs with adjacency matrices $A$ and $B$, and a coupling block $C\in(\mu_4^0)^{q\times q}$, such that
        \[
            H \coloneqq
            \begin{pmatrix}
                A & -C\\
                -C\dagg & B
            \end{pmatrix}
        \]
    has an eigenvector in the synchronized subspace $\subspace$ of the form
        \[ \kpsi = \omegvecvavb, \]
    with eigenvalue $\lambda$, and the other eigenvalue of $H|_{\subspace}$ equals $\lambda+\delta$.

    \begin{proof}
        Since $q$ is even and $k_A,k_B\in\{0,1,\dots,q-1\}$, there exist simple $q$-vertex $k_A$- and $k_B$-regular graphs (existence requires $qk$ to be even and trivially $k < q$), so
            \[ AV_A = k_A V_A, \qquad BV_B = k_B V_B. \]
        Since $l\in\mathcal{L}_q$, \Cref{prop:balanced_discrete_coupling} yields a matrix $C\in(\mu_4^0)^{q\times q}$ such that
            \[ C\allonesq = l\,\allonesq, \qquad C^T\allonesq = l\,\allonesq. \]
        Hence, with $V_A=V_B=\onesqrtq\allonesq$,
            \[ CV_B = l\,V_A, \qquad C\dagg V_A = l^*V_B. \]
        The claim now follows from \Cref{thm:hermitian_arbitrary_state}.
    \end{proof}
\end{corollary}

Thus \Cref{cor:balanced_discrete_hermitian} is the exact discretization step for \Cref{thm:hermitian_arbitrary_state}: whenever the continuous Hermitian design parameters land on the discrete arithmetic lattice (determined by $\mathcal{L}_q, k_A,$ and $k_B$), they are realized by an actual finite weighted graph. In particular, every Gaussian-rational ratio is exactly realizable: if $r=z/w\in\Q(i)^\times$ with $z,w\in\Z[i]\setminus\{0\}$, choose $\lambda=0$ and $\tau=|z|^2$ (equivalently $\delta=|z|^2+|w|^2$), so that $l=\tau/r=w z^* \in\Z[i]$, $k_A=|z|^2$, and $k_B=|w|^2$.
Hence \Cref{cor:balanced_discrete_hermitian} applies for every even integer $q$ satisfying
\begin{align*}
    q\geq& \max\left\{ |z|^2+1,\, |w|^2+1,\, |\operatorname{Re}(wz^*)|+ |\operatorname{Im}(wz^*)| \right\}\\
        =& \max\left\{ k_A+1,\, k_B+1,\, |\operatorname{Re}(l)|+ |\operatorname{Im}(l)| \right\} .
\end{align*}
This gives an explicit sufficient size bound for every Gaussian-rational target ratio (it need not be globally minimal).

\medskip

We now show that, allowing the balanced size $q$ to vary and working modulo global phase, the exact balanced discrete Hermitian realizations form a dense subset of the projective synchronized state space.
Any normalized nonbasis synchronized state is determined, up to global phase, by its ratio
    \[ r \coloneqq \frac{\omega_2}{\omega_1}\in\C^\times. \]
Indeed,
    \[ \kpsi= \omegvecvavb = \omega_1 \begin{pmatrix} V_A\\ rV_B \end{pmatrix}. \]
Thus $\binom{V_A}{rV_B}$ is the unnormalized representative of ratio $r$. Since the two blocks are orthogonal and $\|V_A\|=\|V_B\|=1$,
\[
    \left\|
    \begin{pmatrix}
        V_A\\
        rV_B
    \end{pmatrix}
    \right\|^2
    = \|V_A\|^2 + |r|^2\|V_B\|^2
    = 1+|r|^2.
\]
Therefore normalization forces
    \[ 1=\|\kpsi\|^2=|\omega_1|^2(1+|r|^2), \]
so
    \[ |\omega_1|=\frac{1}{\sqrt{1+|r|^2}}. \]
Writing $\omega_1=e^{i\theta}|\omega_1|$, we obtain
    \[ \kpsi = e^{i\theta}\frac{1}{\sqrt{1+|r|^2}}  \begin{pmatrix} V_A\\ rV_B \end{pmatrix} \]
for some $\theta\in\R$.

\begin{corollary}[Density of exact discrete Hermitian realizations] \label{cor:gaussian_rational_dense}
    Allowing the balanced size $q$ to vary, the set of normalized synchronized states exactly realizable by balanced discrete Hermitian QL-bits is dense, modulo global phase, in the projective synchronized state space $\mathbb{CP}^1$.

    \begin{proof}
        For each balanced size $q$, the synchronized subspace has an ordered orthonormal basis given by its two block-supported synchronized modes. We compare states constructed at different sizes through their coefficient vectors in this basis. Thus every normalized nonbasis synchronized state is represented, up to global phase, by
            \begin{equation*}
                \psi_{\mathrm{red}}(r) \coloneqq \frac{1}{\sqrt{1+|r|^2}} \begin{pmatrix} 1\\ r \end{pmatrix}, \qquad r\in\C^\times.
            \end{equation*}
        
        Since $\Q(i)$ is dense in $\C$, there exists a sequence $r_n\in\Q(i)^\times$ with $r_n\to r$. The map $r\mapsto\psi_{\mathrm{red}}(r)$ is continuous, so
            \[ \psi_{\mathrm{red}}(r_n) \longrightarrow \psi_{\mathrm{red}}(r) \]
        in $\C^2$. By the preceding observation, each $r_n$ is exactly realizable by some balanced discrete Hermitian QL-bit, possibly at a different even size $q_n$, and its synchronized eigenstate has coefficient vector $\psi_{\mathrm{red}}(r_n)$ in the corresponding synchronized basis. Hence the exactly realizable nonbasis states are dense under this common two-level representation.
        
        The basis states are exactly realizable as well, by taking $C=0$ (equivalently $l=0$). Hence the full exactly realizable set is dense, modulo global phase, in the projective synchronized state space $\mathbb{CP}^1$.
    \end{proof}

\end{corollary}

\subsection{A Worked Balanced Example: Exact Realization and Spectral Isolation}\label{sec:discrete_and_density:worked_example}

We close this section with a minimal instance of \Cref{cor:balanced_discrete_hermitian} that also makes the isolation hypothesis of \Cref{cor:hermitian_unique_eigenpair} concrete. Take the Gaussian-rational target
    \[ r = 1+i = \frac{z}{w}, \qquad z=1+i,\quad w=1, \]
for which the instantiation following \Cref{cor:balanced_discrete_hermitian} gives $\lambda=0$, $\delta=|z|^2+|w|^2=3$, $k_A=|z|^2=2$, $k_B=|w|^2=1$, and $l=wz^*=1-i$. Since $l\in\mathcal{L}_4$ and $k_A,k_B\in\{0,\dots,3\}$, the minimal admissible balanced size is $q=4$. Let $P$ denote the cyclic shift on $\Z_4$, and define the Fourier vectors
    \[ \chi_j\coloneqq\frac12\left(1,i^{\,j},i^{2j},i^{3j}\right)^T,\qquad P\chi_j=i^j\chi_j,\qquad j\in\{0,1,2,3\}. \]
Choose
    \[ A=P+P^3 \quad (\text{the $4$-cycle}),\qquad B=P^2 \quad (\text{a perfect matching}), \]
which are simple $2$- and $1$-regular graphs. \Cref{prop:balanced_discrete_coupling} permits any two disjoint matchings to carry the phases of $l=1-i$; we compare two admissible choices,
    \[ C^{(1)} = I - iP \qquad\text{and}\qquad C^{(2)} = I - iP^{2}, \]
both of which have every row and column sum equal to $l$.
Because all four blocks are polynomials in $P$, the full operator $H$ is block-diagonalized by these Fourier vectors: mode $j$ contributes the $2\times2$ Hermitian block
    \[
        \begin{pmatrix}
            a_j & -c_j\\
            -\overline{c_j} & b_j
        \end{pmatrix},
        \qquad a_j = 2\cos(\pi j/2), \qquad b_j = (-1)^{j},
    \]
with coupling symbol $c_j = 1 - i\cdot i^{\,j}$ for $C^{(1)}$ and $c_j = 1 - i\,(-1)^{j}$ for $C^{(2)}$. The mode $j=0$ is the synchronized sector: its block is $\bigl(\begin{smallmatrix} 2 & -(1-i)\\ -(1+i) & 1 \end{smallmatrix}\bigr)$, with eigenvalues $\{0,3\}$ and eigenvector proportional to $(1,\,1+i)^{T}$ at the eigenvalue $0$, exactly as designed. \Cref{tab:worked_example_spectra} lists the three complementary modes.

\begin{table}[t]
    \caption{Complementary Fourier modes of the two balanced $q=4$ realizations of the target $r=1+i$ ($\lambda=0$, $\delta=3$). The synchronized mode $j=0$ contributes the designed eigenvalues $\{0,3\}$ in both cases. For $C^{(1)}$ the coupling symbol vanishes at $j=3$, and the decoupled eigenvalue $0$ collides with the designed $\lambda=0$; for $C^{(2)}$ no symbol vanishes and the designed pair is isolated.}
    \label{tab:worked_example_spectra}
    \centering
    \begin{tabular}{cccccc}
        $j$ & $(a_j,b_j)$ & $c_j$ for $C^{(1)}$ & eigenvalues & $c_j$ for $C^{(2)}$ & eigenvalues \\
        \hline
        $1$ & $(0,-1)$ & $2$ & $\tfrac{-1\pm\sqrt{17}}{2}$ & $1+i$ & $1,\ -2$ \\
        $2$ & $(-2,1)$ & $1+i$ & $\tfrac{-1\pm\sqrt{17}}{2}$ & $1-i$ & $\tfrac{-1\pm\sqrt{17}}{2}$ \\
        $3$ & $(0,-1)$ & $0$ & $0,\ -1$ & $1+i$ & $1,\ -2$ \\
    \end{tabular}
\end{table}

For $C^{(1)}$, the two phases cancel on the mode $j=3$: $c_3 = 1 - i\cdot i^{3} = 0$. That mode decouples and contributes the eigenvalues $\{0,-1\}$, so $(\chi_3,\mathbf{0}_4)^{T}$ is an eigenvector of $H$ with eigenvalue $0$ supported entirely on $G_A$. The designed eigenvalue $\lambda=0$ therefore has multiplicity two in $H$, the hypothesis of \Cref{cor:hermitian_unique_eigenpair} fails, and the designed eigenpair, while still present by \Cref{cor:balanced_discrete_hermitian}, is no longer uniquely addressable by its eigenvalue. For $C^{(2)}$ no coupling symbol vanishes; the complementary spectrum is $\bigl\{1,\,1,\,-2,\,-2,\,\tfrac{-1+\sqrt{17}}{2},\,\tfrac{-1-\sqrt{17}}{2}\bigr\}$, disjoint from $\{0,3\}$ (the minimal separations are $1$ and $\tfrac{7-\sqrt{17}}{2}\approx 1.44$), so \Cref{cor:hermitian_unique_eigenpair} applies: both designed eigenvalues are simple and their eigenvectors lie in $\subspace$.

This example shows that the effective parameters $(k_A,k_B,l)$ alone do not determine full spectral isolation. Both coupling blocks realize the same synchronized $2\times2$ operator, but their actions on $\subspace^\perp$ differ: one creates a collision through Fourier-mode cancellation, while the other avoids it. The matching freedom in \Cref{prop:balanced_discrete_coupling} can therefore be used not only to realize the desired effective coupling, but also to shape the complementary spectrum. For less symmetric instances, the condition in \Cref{cor:hermitian_unique_eigenpair} remains a finite spectral check.

\section{Conclusion}

We have characterized the exact synchronized two-level states that can be realized by block-coupled regular graph operators with complex edge weights. The central result is that complex amplitudes are not obtained merely by allowing complex parameters in a symmetric reduced block: under a real synchronized-sector spectrum requirement, complex-symmetric coupling and real symmetric coupling with complex detuning obey the phase constraints $r^2\in\R$ and $r+r\inv\in\R$, respectively. Hermitian conjugate pairing removes these constraints. In the Hermitian model, every nonbasis target ratio $r\in\C^\times$ can be realized with an arbitrary prescribed real eigenvalue and nonzero signed spectral gap, and the synchronized subspace is a reducing subspace of the full operator.

The asymmetric construction shows that universality can also be obtained by independently tuning the two coupling directions, although without the automatic orthogonal decoupling and spectral guarantees supplied by Hermitian symmetry. Finally, the balanced $\mu_4^0=\mufz$ construction gives a strong finite-graph counterpart to the continuous universality result: every Gaussian-rational amplitude ratio has an exact finite weighted-graph realization, and these exact realizations are dense, modulo global phase, in the projective synchronized pure-state space. Thus, the expressive power of complex QL-bits is governed by the structural relation between reciprocal couplings, conjugation symmetry, and arithmetic row-column regularity,
with Hermitian pairing providing the most robust mechanism for universal real-spectrum state synthesis.
\section*{Acknowledgments}

S.K. would like to acknowledge the support of the U.S. Department of Energy (DOE) Office of Basic Energy Sciences, under grant number DE-SC0026309. The authors declare no competing interests. %
\appendix

\section{Structural consequences of the Hermitian construction}\label{app:herm_structural_consequences_proofs}

\begin{proof}[Proof of \Cref{prop:hermitian_closed_subsystem}]
    The invariance of $\subspace$ follows directly from \Cref{prop:general_synchronized_reduction}.

    Now let
        \[ u=\begin{pmatrix}x\\y\end{pmatrix}\in \subspace^\perp, \]
    so that $x\perp V_A$ and $y\perp V_B$. Then
        \[ Hu= \begin{pmatrix} Ax-Cy\\ -C\dagg x+By \end{pmatrix}. \]
    To show that $Hu\in \subspace^\perp$, it suffices to check that its upper component is orthogonal to $V_A$ and its lower component is orthogonal to $V_B$. Using that $A$ and $B$ are Hermitian and that $AV_A=k_A V_A$, $BV_B=k_B V_B$, we compute
    \[
        \braket{V_A}{Ax-Cy}
        = \braket{AV_A}{x} - \braket{V_A}{Cy}
        =  k_A \braket{V_A}{x} - \braket{C\dagg V_A}{y}
        =  k_A \braket{V_A}{x} -l \braket{V_B}{y} %
        = 0,
    \]
    and similarly
    \[
        \braket{V_B}{-C\dagg x+By}
        = -\braket{CV_B}{x} + \braket{BV_B}{y}
        = -l^*\braket{V_A}{x} + k_B\braket{V_B}{y}
        = 0.
    \]

    Hence $Hu\in \subspace^\perp$. Therefore both $\subspace$ and $\subspace^\perp$ are invariant under $H$, and so $H$ is block diagonal with respect to the decomposition
        \[ \C^{\,n+m}=\subspace\oplus \subspace^\perp. \]
\end{proof}

\begin{proof}[Proof of \Cref{cor:hermitian_unique_eigenpair}]
    By \Cref{thm:hermitian_arbitrary_state}, the restriction $H|_\subspace$ has eigenvalues $\lambda$ and $\lambda+\delta$ with $\delta\neq 0$, so these two eigenvalues are distinct.
    
    By \Cref{prop:hermitian_closed_subsystem},
        \[ H = H|_\subspace \oplus H|_{\subspace^\perp}, \]
    and therefore
        \[ \sigma(H)=\sigma(H|_\subspace)\cup \sigma(H|_{\subspace^\perp}). \]
    Hence $\lambda$ and $\lambda+\delta$ are eigenvalues of the full matrix $H$.
    
    Now let $\mu\in\{\lambda,\lambda+\delta\}$, and let $z$ be an eigenvector of $H$ with eigenvalue $\mu$. Write
        \[ z=s+t, \qquad s\in \subspace,\quad t\in \subspace^\perp. \]
    Since the decomposition is block diagonal,
        \[ Hz=H|_\subspace s+H|_{\subspace^\perp} t. \]
    On the other hand,
        \[ Hz=\mu z=\mu s+\mu t. \]
    By uniqueness of the decomposition into $\subspace\oplus \subspace^\perp$, we obtain
        \[ H|_\subspace s=\mu s, \qquad H|_{\subspace^\perp} t=\mu t. \]
    But $\mu\notin \sigma(H|_{\subspace^\perp})$ by assumption, so the second equation forces $t=0$. Thus $z=s\in \subspace$. Therefore every eigenvector of $H$ with eigenvalue $\lambda$ or $\lambda+\delta$ lies entirely in $\subspace$.
    
    Finally, each of the two eigenvalues is simple on the $2\times 2$ block $H|_\subspace$, because they are distinct, and neither appears in $H|_{\subspace^\perp}$. Hence each is simple for the full matrix $H$.
\end{proof} 

\bibliographystyle{ACM-Reference-Format}%
\bibliography{refs}%

@article{dickey2025construction,
    title={Programmable Quantum-Like bits from Signed Regular Graphs},
    author={Dickey, Ethan and Vyas, Abhijeet and Kais, Sabre},
    journal={Royal Society Open Science},
    year={2026},
    note={Accepted for publication},
    eprint={2507.21289},
    archivePrefix={arXiv},
    primaryClass={quant-ph},
    url={https://arxiv.org/abs/2507.21289}
}

@article{pillai2005perron,
    title={The Perron-Frobenius theorem: some of its applications},
    author={Pillai, S Unnikrishna and Suel, Torsten and Cha, Seunghun},
    journal={IEEE Signal Processing Magazine},
    volume={22},
    number={2},
    pages={62--75},
    year={2005},
    publisher={IEEE},
    doi = {10.1109/MSP.2005.1406483} 
}

@article{scholes2024quantumlike,
    author = {Scholes, Gregory D.},
    title = {Quantum-like states on complex synchronized networks},
    journal = {Proceedings of the Royal Society A: Mathematical, Physical and Engineering Sciences},
    volume = {480},
    number = {2295},
    pages = {20240209},
    year = {2024},
    month = {08},
    issn = {1364-5021},
    doi = {10.1098/rspa.2024.0209},
    url = {https://doi.org/10.1098/rspa.2024.0209},
    eprint = {https://royalsocietypublishing.org/rspa/article-pdf/doi/10.1098/rspa.2024.0209/512921/rspa.2024.0209.pdf},
}

@article{amati2025quantumlikebits,
    title = {Quantum information with quantumlike bits},
    author = {Amati, Graziano and Scholes, Gregory D.},
    journal = {Phys. Rev. A},
    volume = {111},
    issue = {6},
    pages = {062203},
    numpages = {15},
    year = {2025},
    month = {Jun},
    publisher = {American Physical Society},
    doi = {10.1103/PhysRevA.111.062203},
    url = {https://link.aps.org/doi/10.1103/PhysRevA.111.062203}
}

@article{scholes2025productstates,
    title = {Quantumlike Product States Constructed from Classical Networks},
    author = {Scholes, Gregory D. and Amati, Graziano},
    journal = {Phys. Rev. Lett.},
    volume = {134},
    issue = {6},
    pages = {060202},
    numpages = {6},
    year = {2025},
    month = {Feb},
    publisher = {American Physical Society},
    doi = {10.1103/PhysRevLett.134.060202},
    url = {https://link.aps.org/doi/10.1103/PhysRevLett.134.060202}
}

@article{bravyi2005universal,
    title = {Universal quantum computation with ideal Clifford gates and noisy ancillas},
    author = {Bravyi, Sergey and Kitaev, Alexei},
    journal = {Phys. Rev. A},
    volume = {71},
    issue = {2},
    pages = {022316},
    numpages = {14},
    year = {2005},
    month = {Feb},
    publisher = {American Physical Society},
    doi = {10.1103/PhysRevA.71.022316},
    url = {https://link.aps.org/doi/10.1103/PhysRevA.71.022316}
}

@article{howard2014contextuality,
    author    = {Mark Howard and Joel Wallman and Victor Veitch and Joseph Emerson},
    title     = {Contextuality supplies the ‘magic’ for quantum computation},
    journal   = {Nature},
    year      = {2014},
    volume    = {510},
    number    = {7505},
    pages     = {351--355},
    month     = jun,
    doi = {10.1038/nature13460},
}

@article{veitch2014resource,
    title = {The resource theory of stabilizer quantum computation},
    author = {Veitch, Victor and Hamed Mousavian, S A and Gottesman, Daniel and Emerson, Joseph},
    doi = {10.1088/1367-2630/16/1/013009},
    url = {https://doi.org/10.1088/1367-2630/16/1/013009},
    year = {2014},
    month = {jan},
    publisher = {IOP Publishing},
    volume = {16},
    number = {1},
    pages = {013009},
    journal = {New Journal of Physics},
}

@article{bravyi2012lowoverhead,
    title = {Magic-state distillation with low overhead},
    author = {Bravyi, Sergey and Haah, Jeongwan},
    journal = {Phys. Rev. A},
    volume = {86},
    issue = {5},
    pages = {052329},
    numpages = {10},
    year = {2012},
    month = {Nov},
    publisher = {American Physical Society},
    doi = {10.1103/PhysRevA.86.052329},
    url = {https://link.aps.org/doi/10.1103/PhysRevA.86.052329}
}

@article{jones2013multilevel,
    title = {Multilevel distillation of magic states for quantum computing},
    author = {Jones, Cody},
    journal = {Phys. Rev. A},
    volume = {87},
    issue = {4},
    pages = {042305},
    numpages = {8},
    year = {2013},
    month = {Apr},
    publisher = {American Physical Society},
    doi = {10.1103/PhysRevA.87.042305},
    url = {https://link.aps.org/doi/10.1103/PhysRevA.87.042305}
}

@article{campbell2017roads,
    author  = {Campbell, Earl T. and Terhal, Barbara M. and Vuillot, Christophe},
    title   = {Roads towards fault-tolerant universal quantum computation},
    journal = {Nature},
    year    = {2017},
    volume  = {549},
    number  = {7671},
    pages   = {172--179},
    month   = sep,
    doi     = {10.1038/nature23460},
    url     = {https://doi.org/10.1038/nature23460},
    issn    = {1476-4687}
}

@article{seddon2019quantifying,
    author = {Seddon, James R. and Campbell, Earl T.},
    title = {Quantifying magic for multi-qubit operations},
    journal = {Proceedings of the Royal Society A: Mathematical, Physical and Engineering Sciences},
    volume = {475},
    number = {2227},
    pages = {20190251},
    year = {2019},
    month = {07},
    issn = {1364-5021},
    doi = {10.1098/rspa.2019.0251},
    url = {https://doi.org/10.1098/rspa.2019.0251},
    eprint = {https://royalsocietypublishing.org/rspa/article-pdf/doi/10.1098/rspa.2019.0251/1221292/rspa.2019.0251.pdf},
}

@article{veitch2013efficient,
    doi = {10.1088/1367-2630/15/1/013037},
    url = {https://doi.org/10.1088/1367-2630/15/1/013037},
    year = {2013},
    month = {jan},
    publisher = {IOP Publishing},
    volume = {15},
    number = {1},
    pages = {013037},
    author = {Veitch, Victor and Wiebe, Nathan and Ferrie, Christopher and Emerson, Joseph},
    title = {Efficient simulation scheme for a class of quantum optics experiments with non-negative Wigner representation},
    journal = {New Journal of Physics},
}

@article{bravyi2016trading,
    title = {Trading Classical and Quantum Computational Resources},
    author = {Bravyi, Sergey and Smith, Graeme and Smolin, John A.},
    journal = {Phys. Rev. X},
    volume = {6},
    issue = {2},
    pages = {021043},
    numpages = {14},
    year = {2016},
    month = {Jun},
    publisher = {American Physical Society},
    doi = {10.1103/PhysRevX.6.021043},
    url = {https://link.aps.org/doi/10.1103/PhysRevX.6.021043}
}

@article{farhi1998quantum,
    title = {Quantum computation and decision trees},
    author = {Farhi, Edward and Gutmann, Sam},
    journal = {Phys. Rev. A},
    volume = {58},
    issue = {2},
    pages = {915--928},
    numpages = {0},
    year = {1998},
    month = {Aug},
    publisher = {American Physical Society},
    doi = {10.1103/PhysRevA.58.915},
    url = {https://link.aps.org/doi/10.1103/PhysRevA.58.915}
}

@article{childs2009universal,
    title = {Universal Computation by Quantum Walk},
    author = {Childs, Andrew M.},
    journal = {Phys. Rev. Lett.},
    volume = {102},
    issue = {18},
    pages = {180501},
    numpages = {4},
    year = {2009},
    month = {May},
    publisher = {American Physical Society},
    doi = {10.1103/PhysRevLett.102.180501},
    url = {https://link.aps.org/doi/10.1103/PhysRevLett.102.180501}
}

@article{bose2003communication,
    title = {Quantum Communication through an Unmodulated Spin Chain},
    author = {Bose, Sougato},
    journal = {Phys. Rev. Lett.},
    volume = {91},
    issue = {20},
    pages = {207901},
    numpages = {4},
    year = {2003},
    month = {Nov},
    publisher = {American Physical Society},
    doi = {10.1103/PhysRevLett.91.207901},
    url = {https://link.aps.org/doi/10.1103/PhysRevLett.91.207901}
}

@article{christandl2004perfect,
    title = {Perfect State Transfer in Quantum Spin Networks},
    author = {Christandl, Matthias and Datta, Nilanjana and Ekert, Artur and Landahl, Andrew J.},
    journal = {Phys. Rev. Lett.},
    volume = {92},
    issue = {18},
    pages = {187902},
    numpages = {4},
    year = {2004},
    month = {May},
    publisher = {American Physical Society},
    doi = {10.1103/PhysRevLett.92.187902},
    url = {https://link.aps.org/doi/10.1103/PhysRevLett.92.187902}
}

@article{christandl2005perfect,
    title = {Perfect transfer of arbitrary states in quantum spin networks},
    author = {Christandl, Matthias and Datta, Nilanjana and Dorlas, Tony C. and Ekert, Artur and Kay, Alastair and Landahl, Andrew J.},
    journal = {Phys. Rev. A},
    volume = {71},
    issue = {3},
    pages = {032312},
    numpages = {11},
    year = {2005},
    month = {Mar},
    publisher = {American Physical Society},
    doi = {10.1103/PhysRevA.71.032312},
    url = {https://link.aps.org/doi/10.1103/PhysRevA.71.032312}
}

@article{kempe2003quantum,
    author = {Julia Kempe},
    title = {Quantum random walks: An introductory overview},
    journal = {Contemporary Physics},
    volume = {44},
    number = {4},
    pages = {307--327},
    year = {2003},
    publisher = {Taylor \& Francis},
    doi = {10.1080/00107151031000110776},
    URL = {https://doi.org/10.1080/00107151031000110776},
    eprint = {https://doi.org/10.1080/00107151031000110776}
}

@article{krovi2007quantum,
    title = {Quantum walks on quotient graphs},
    author = {Krovi, Hari and Brun, Todd A.},
    journal = {Phys. Rev. A},
    volume = {75},
    issue = {6},
    pages = {062332},
    numpages = {14},
    year = {2007},
    month = {Jun},
    publisher = {American Physical Society},
    doi = {10.1103/PhysRevA.75.062332},
    url = {https://link.aps.org/doi/10.1103/PhysRevA.75.062332}
}

@article{godsil2012state,
    author = {Chris Godsil},
    title = {State transfer on graphs},
    journal = {Discrete Mathematics},
    volume = {312},
    number = {1},
    pages = {129-147},
    year = {2012},
    note = {Algebraic Graph Theory — A Volume Dedicated to Gert Sabidussi on the Occasion of His 80th Birthday},
    issn = {0012-365X},
    doi = {https://doi.org/10.1016/j.disc.2011.06.032},
    url = {https://www.sciencedirect.com/science/article/pii/S0012365X11002974},
}

@article{bachman2012perfect,
    author = {Rachel Bachman and Eric Fredette and Jessica Fuller and Michael Landry and Michael C. Opperman and Christino Tamon and Andrew Tollefson},
    title = {Perfect state transfer on quotient graphs},
    journal = {Quantum Information \& Computation},
    volume = {12},
    number = {3-4},
    pages = {293--313},
    year = {2012},
    note = {Published version URL unavailable; arXiv: \url{https://arxiv.org/abs/1108.0339}}
}

@article{aguiar2018synchronization,
    author = {Aguiar, Manuela A. D. and Dias, Ana Paula S.},
    title = {Synchronization and equitable partitions in weighted networks},
    journal = {Chaos: An Interdisciplinary Journal of Nonlinear Science},
    volume = {28},
    number = {7},
    pages = {073105},
    year = {2018},
    month = {07},
    issn = {1054-1500},
    doi = {10.1063/1.4997385},
    url = {https://doi.org/10.1063/1.4997385}
}

@article{chu2002structured,
    author={Chu, Moody T. and Golub, Gene H.},
    title={Structured inverse eigenvalue problems},
    volume={11},
    DOI={10.1017/S0962492902000016},
    journal={Acta Numerica},
    year={2002},
    pages={1–71}
}

@article{bender1998real,
    title = {Real Spectra in Non-Hermitian Hamiltonians Having $\mathcal{P}\mathcal{T}$ Symmetry},
    author = {Bender, Carl M. and Boettcher, Stefan},
    journal = {Phys. Rev. Lett.},
    volume = {80},
    issue = {24},
    pages = {5243--5246},
    numpages = {0},
    year = {1998},
    month = {Jun},
    publisher = {American Physical Society},
    doi = {10.1103/PhysRevLett.80.5243},
    url = {https://link.aps.org/doi/10.1103/PhysRevLett.80.5243}
}

@article{bender2002complex,
    title = {Complex Extension of Quantum Mechanics},
    author = {Bender, Carl M. and Brody, Dorje C. and Jones, Hugh F.},
    journal = {Phys. Rev. Lett.},
    volume = {89},
    issue = {27},
    pages = {270401},
    numpages = {4},
    year = {2002},
    month = {Dec},
    publisher = {American Physical Society},
    doi = {10.1103/PhysRevLett.89.270401},
    url = {https://link.aps.org/doi/10.1103/PhysRevLett.89.270401}
}

@article{mostafazadeh2002pseudo,
    author = {Mostafazadeh, Ali},
    title = {Pseudo-Hermiticity versus PT symmetry: The necessary condition for the reality of the spectrum of a non-Hermitian Hamiltonian},
    journal = {Journal of Mathematical Physics},
    volume = {43},
    number = {1},
    pages = {205-214},
    year = {2002},
    month = {01},
    issn = {0022-2488},
    doi = {10.1063/1.1418246},
    url = {https://doi.org/10.1063/1.1418246}
}

@article{elganainy2018nonhermitian,
    author  = {El-Ganainy, Ramy and Makris, Konstantinos G. and Khajavikhan, Mercedeh and Musslimani, Ziad H. and Rotter, Stefan and Christodoulides, Demetrios N.},
    title   = {Non-Hermitian physics and PT symmetry},
    journal = {Nature Physics},
    year    = {2018},
    volume  = {14},
    number  = {1},
    pages   = {11--19},
    doi     = {10.1038/nphys4323},
    url     = {https://doi.org/10.1038/nphys4323},
    issn    = {1745-2481},
}

@article{ashida2020nonhermitian,
    author = {Yuto Ashida and Zongping Gong and Masahito Ueda},
    title = {Non-Hermitian physics},
    journal = {Advances in Physics},
    volume = {69},
    number = {3},
    pages = {249--435},
    year = {2020},
    publisher = {Taylor \& Francis},
    doi = {10.1080/00018732.2021.1876991},
    URL = {https://doi.org/10.1080/00018732.2021.1876991},
    eprint = {https://doi.org/10.1080/00018732.2021.1876991}
}

@article{bergholtz2021exceptional,
    title = {Exceptional topology of non-Hermitian systems},
    author = {Bergholtz, Emil J. and Budich, Jan Carl and Kunst, Flore K.},
    journal = {Rev. Mod. Phys.},
    volume = {93},
    issue = {1},
    pages = {015005},
    numpages = {31},
    year = {2021},
    month = {Feb},
    publisher = {American Physical Society},
    doi = {10.1103/RevModPhys.93.015005},
    url = {https://link.aps.org/doi/10.1103/RevModPhys.93.015005}
}

@book{quarteroni2014reduced,
    editor    = {Alfio Quarteroni and Gianluigi Rozza},
    title     = {Reduced Order Methods for Modeling and Computational Reduction},
    series    = {MS\&A -- Modeling, Simulation and Applications},
    volume    = {9},
    year      = {2014},
    publisher = {Springer International Publishing},
    address   = {Cham},
    doi       = {10.1007/978-3-319-02090-7},
    isbn      = {978-3-319-02089-1}
}

@article{benner2015survey,
    author = {Benner, Peter and Gugercin, Serkan and Willcox, Karen},
    title = {A Survey of Projection-Based Model Reduction Methods for Parametric Dynamical Systems},
    journal = {SIAM Review},
    volume = {57},
    number = {4},
    pages = {483-531},
    year = {2015},
    doi = {10.1137/130932715},
    URL = {https://doi.org/10.1137/130932715},
    eprint = {https://doi.org/10.1137/130932715}
}

@ARTICLE{shuman2013emerging,
    author={Shuman, David I and Narang, Sunil K. and Frossard, Pascal and Ortega, Antonio and Vandergheynst, Pierre},
    journal={IEEE Signal Processing Magazine}, 
    title={The emerging field of signal processing on graphs: Extending high-dimensional data analysis to networks and other irregular domains}, 
    year={2013},
    volume={30},
    number={3},
    pages={83-98},
    doi={10.1109/MSP.2012.2235192}
}

@ARTICLE{ortega2018graph,
    author={Ortega, Antonio and Frossard, Pascal and Kovačević, Jelena and Moura, José M. F. and Vandergheynst, Pierre},
    journal={Proceedings of the IEEE}, 
    title={Graph Signal Processing: Overview, Challenges, and Applications}, 
    year={2018},
    volume={106},
    number={5},
    pages={808-828},
    doi={10.1109/JPROC.2018.2820126}
}

@article{zaslavsky1982signed,
    title = {Signed graphs},
    journal = {Discrete Applied Mathematics},
    volume = {4},
    number = {1},
    pages = {47-74},
    year = {1982},
    issn = {0166-218X},
    doi = {https://doi.org/10.1016/0166-218X(82)90033-6},
    url = {https://www.sciencedirect.com/science/article/pii/0166218X82900336},
    author = {Thomas Zaslavsky}
}

@article{reff2012spectral,
    title = {Spectral properties of complex unit gain graphs},
    journal = {Linear Algebra and its Applications},
    volume = {436},
    number = {9},
    pages = {3165-3176},
    year = {2012},
    issn = {0024-3795},
    doi = {https://doi.org/10.1016/j.laa.2011.10.021},
    url = {https://www.sciencedirect.com/science/article/pii/S002437951100718X},
    author = {Nathan Reff}
}

@article{lieb1993fluxes,
    author = {Lieb, Elliott H. and Loss, Michael},
    title = {Fluxes, {L}aplacians, and {K}asteleyn's theorem},
    journal = {Duke Mathematical Journal},
    volume = {71},
    number = {2},
    year = {1993},
    pages = {337--363},
    doi = {10.1215/S0012-7094-93-07114-1}
}

\end{document}